\documentclass[preprint,12pt,authoryear]{elsarticle}
\usepackage{xcolor}
\usepackage{booktabs}
\usepackage{tikz}
\usetikzlibrary{arrows.meta, positioning, shapes.geometric}
\definecolor{RoyalBlue}{RGB}{65,105,225}
\definecolor{CoralRed}{RGB}{205,92,92}
\definecolor{DarkGray}{RGB}{80,80,80}
\definecolor{LightGray}{RGB}{220,220,220}
\newcommand{\rev}[1]{#1} % CLEAN/final version: revision markup neutralized (prints normal black). For the marked-up version, restore: \newcommand{\rev}[1]{\textcolor{RoyalBlue}{#1}}

\usepackage{amssymb}
\usepackage{amsmath}
\usepackage{amsthm}   % <-- must be active!

\theoremstyle{definition}
\newtheorem{theorem}{Theorem}[section]
\newtheorem{proposition}{Proposition}[section]
\newtheorem{lemma}{Lemma}[section]
\newtheorem{corollary}[theorem]{Corollary} % <-- add this
\newtheorem{assumption}{Assumption}
\newtheorem{definition}{Definition}[section]

\theoremstyle{remark}
\newtheorem{remark}{Remark}

\newtheorem*{proposition*}{Proposition}
\newtheorem*{theorem*}{Theorem}
\usepackage{tikz}
\usetikzlibrary{positioning, shapes.geometric, arrows.meta}
\usepackage[hidelinks]{hyperref}
\begin{document}

\begin{frontmatter}



\title{Universalization and the Origins of Fiscal Capacity}

\author[inst1,inst2]{Esteban Muñoz-Sobrado\corref{cor1}}
\ead{esteban.munoz@urv.cat}

\affiliation[inst1]{%
organization={Department of Economics, Universitat Rovira i Virgili},
addressline={Av. de la Universitat 1},
city={Reus},
postcode={43204},
country={Spain}}

\affiliation[inst2]{%
organization={ECO-SOS, Universitat Rovira i Virgili},
city={Tarragona},
country={Spain}}

\cortext[cor1]{Corresponding author. Email: esteban.munoz@urv.cat. I am grateful to Ingela Alger, Pau Juan-Bartroli, and Enrico Mattia Salonia for insightful comments and suggestions, and to the editor, an associate editor, and two anonymous referees for comments that substantially improved the paper. This work was supported by the Spanish Ministry of Science, Innovation and Universities [grant number PID2022-137382NB-I00].}

%% Abstract
\begin{abstract}
This paper proposes a model of tax compliance and fiscal 
capacity in which citizens partially internalize the 
consequences of concealment by imagining a world in which 
a share of the population acted similarly, linking their compliance decisions 
to the perceived value of public spending.
A selfish elite chooses between public goods and private rents, taking compliance as given. 
In equilibrium, citizens’ moral internalization expands the feasible tax base and induces elites to allocate resources toward provision rather than appropriation. 
When the value of public spending is uncertain, morality enables credible reform: high-value elites can signal their type through provision, prompting citizens to increase compliance and raising fiscal capacity within the same period. 
The analysis thus identifies a moral channel through which states may escape low-capacity traps even under weak enforcement.
\end{abstract}

\begin{graphicalabstract}
\end{graphicalabstract}

%% Research highlights
%\item Citizens’ ethics microfound a moral Laffer curve for fiscal capacity.
%\item State capacity emerges endogenously from interactions between elite and citizens.
%\item Moral preferences expand the feasible tax base and sustain compliance.
%\item Citizens’ ethics enable credible and self-enforcing fiscal reform.
%\item Provides a contractual theory of state formation grounded in ethics.


\end{frontmatter}



\section{Introduction} \label{s1}

What sustains tax compliance and state capacity when enforcement is imperfect and institutions are fragile? 
Much of the economic literature emphasizes a coercive view of the state, where capacity depends on investments that extend its reach and allow it to enforce compliance \citep{besley2009origins,besley2011pillars}.
In contrast, a second view considers the emergence of a strong state as a collective arrangement in which citizens voluntarily comply with taxes in exchange for public goods and institutional legitimacy \citep{weingast1997political,weingast2005constitutional,binmore1994game,binmore1998evolution,ACEMOGLU20051199}. 
Within this view, economic models have microfounded compliance through intrinsic reciprocity, where reciprocal behavior between the elite and citizens is driven by citizens' preferences \citep{besley2020}. 
Tax morale in such frameworks depends on perceptions of government behavior and can sustain either high- or low-capacity equilibria depending on expectations about the expenditure mix of the government.

This paper formalizes a different mechanism based on universalization reasoning, the idea that individuals evaluate their actions by asking, \textit{What if everyone acted as I do?} 
We capture this logic using Homo Moralis preferences 
\citep{alger2013homo,alger2016evolution}, under which 
individuals evaluate outcomes as if a hypothetical share of the population 
acted the same way. The higher this share, the stronger is 
moral internalization. Full universalization is the limiting 
case. These preferences admit an evolutionary microfoundation but can be studied independently of it. In the context of taxation, this approach implies that individuals may comply even when enforcement is limited, because deviation carries an intrinsic moral cost as citizens internalize the universalized payoff implied by their compliance. 
The model shows that moral internalization can expand the fiscal frontier and foster the provision of public goods even under weak institutions. 
Moreover, when moral internalization is strong enough, credible improvements in fiscal behavior become feasible: self-interested elites find it optimal to allocate revenue to public goods rather than private rents, and citizens respond with higher compliance, thereby increasing fiscal capacity. 
By linking universalization reasoning to fiscal behavior, the analysis provides a microfoundation for self-enforcing taxation and offers a novel explanation of how states can move out of persistent low-capacity traps.

We build on the framework of \citet{besley2020}, where an elite 
sets the tax rate and allocates revenue between public goods and 
self-serving transfers. The key departure lies in the behavioral 
foundation of compliance. In Besley's model, reciprocity enters 
through the composition of spending: civic-minded citizens 
comply more when the government allocates resources to public 
goods rather than rents, and civic-mindedness is a fixed trait. 
In our framework, compliance is driven by universalization: the 
moral return to reporting depends on the value of public 
spending $\alpha$ and institutional quality $\sigma$, not only 
on the spending mix. This structural dependence generates 
distinct predictions, particularly under asymmetric information, 
where universalization creates an informational channel absent 
under reciprocity.

We first explore these mechanisms in a static model with full 
information, establishing how morality shapes fiscal capacity 
and elite incentives, and then extend the framework to a dynamic 
setting in which the elite privately observes the value of public 
spending, so that its allocation choice serves as a signal and 
compliance responds to inferred fundamentals.

The static model delivers two main results.
First, higher morality increases tax compliance, expanding the 
feasible revenue set and encouraging elites to allocate 
resources toward public goods rather than private rents. 
This mechanism supports high-tax, high-capacity equilibria. 
Second, moral preferences can sustain public provision even 
when the value of public goods is low, because sufficiently 
strong morality shifts the elite's incentive constraint in 
favor of provision. In this region, morality effectively 
disciplines elites into providing public goods that they would 
otherwise forgo, a role analogous to the cooperation incentives 
in \citet{ACEMOGLU20051199} but operating through moral 
internalization rather than repeated interaction.

To take the analysis beyond this static benchmark, we then study an infinite-horizon environment in which the value of public spending evolves over time and is privately observed by the elite. The elite's allocation choice becomes a signal of fiscal fundamentals: when citizens observe public spending rather than rents, they update their beliefs about the value of public goods. Under universalization, this update translates into higher compliance within the same period. This informational channel is absent under reciprocity, where compliance responds to the observed spending mix but not to inferred fundamentals. The dynamic structure thus produces a distinctive prediction, which we call \emph{credible reform}: a self-enforcing equilibrium in which elite provision reveals a high-value state and the resulting belief update triggers an immediate expansion of fiscal capacity.

The dynamic analysis yields two main results. 
First, morality activates credible reform: when the value of 
public goods is only moderately high, self-interested elites 
would not provide them on material grounds alone, but 
strong enough moral internalization makes provision a credible signal 
of high productivity, prompting citizens to raise compliance 
within the same period. 
Second, morality amplifies the fiscal response to reform: 
higher morality magnifies the compliance jump triggered by 
credible provision, expanding the tax base immediately.

In our analysis, we distinguish between a weak-high state, in which morality is essential for credible reform, and a strong-high state, in which strong provision incentives make reform self-enforcing. 
Together, these findings suggest that moral societies have a 
structural advantage in escaping low-capacity equilibria once 
they experience positive shocks to the value of their common 
goods, an implication that does not follow directly from 
reciprocity-based models. We illustrate these mechanisms with 
two cases: the Nordic countries, where high moral internalization 
sustains compliance under moderate enforcement 
\citep{rothstein2005all,uslaner2008corruption}, and Botswana, 
where pre-existing civic norms enabled credible fiscal reform 
following the discovery of diamonds 
\citep{acemoglu2003african,robinson2006political}.

\paragraph*{Related literature}
This paper relates to a broad literature that examines how cooperation and state authority can be sustained without perfect commitment. 
For instance, \citet{kotlikoff1988social} study self-enforcing intergenerational social contracts in an overlapping-generations framework, while \citet{binmore1998evolution} models the endogenous formation of fairness norms through repeated interaction. 
\citet{ACEMOGLU20051199} develops a related framework in which a consensually strong state capable of taxation and rule enforcement emerges when the gains from cooperation outweigh temptations to expropriate. 
In these approaches, compliance and state capacity arise from strategic or institutional consistency through repeated play, reputation, or procedural rules that make cooperation incentive compatible. 
By contrast, the mechanism developed here does not rely on repetition or belief-based enforcement. 
It derives compliance from moral internalization, grounded in 
universalization reasoning, where individuals behave as if their 
own actions were to be generalized. This provides a 
belief-independent foundation for voluntary compliance and the 
emergence of fiscal capacity, complementing existing 
reciprocity-based explanations.

The paper also connects to work emphasizing the role of cultural 
traits such as trust and shared identity in shaping institutional 
performance \citep{algan2013trust, collier2017culture}, and more 
broadly to the literature on the positive effects of generalized 
trust on economic growth \citep{tabellini2010culture}, and to 
research on the interaction between economic incentives and social 
norms in redistributive contexts \citep{lindbeck1999social}. 
These accounts typically model civic culture as generalized 
trust, in-group identification, or reciprocity norms. Our 
approach instead treats moral universalization as a preference 
that enters individual optimization directly, without requiring 
beliefs about others' compliance. 
Section~\ref{sec:robustness} shows how the population share of 
this preference can itself be sustained endogenously through 
cultural transmission.

Finally, we contribute to a literature that applies Kantian 
reasoning to economic behavior, building on the early insight 
of \citet{laffont1975macroeconomic}. Specifically, we draw on 
the Homo Moralis framework of 
\citet{alger2013homo, alger2016evolution}, in which 
individuals evaluate their actions as if universally adopted. 
This framework has been applied to settings such as team 
incentives \citep{sarkisian2017team, sarkisian2021optimal, 
sarkisian2021screening}, collective choice 
\citep{algerlaslier2022, algerdierkslaslier2025}, climate policy 
\citep{eichner2020climate}, Pigovian 
taxation \citep{eichner2020kantians}, monetary institutions 
\citep{norman2020evolution}, moral preferences in bargaining 
\citep{juan2024moral}, injunctive norms 
\citep{juan2024injunctive}, trade under information 
asymmetries \citep{rivero2025trade}, norm emergence 
\citep{alger2025norms}, and the interaction between moral 
concerns and social preferences 
\citep{juanbartroli2025socialpreferencesmoralconcerns}\footnote{A companion paper, \citet{munoz2022taxing}, applies 
{Homo Moralis} preferences to optimal income taxation. 
The present paper focuses instead on the emergence and 
persistence of fiscal capacity. \citet{lopezcantor2025individual} 
analyzes the interplay between individual morality, public good 
provision, and redistribution in a static political economy 
setting, complementing the present analysis by examining the 
distributional consequences of Kantian moral preferences through 
voting equilibria.}.

The remainder of the paper is organized as follows. 
Section~\ref{sec:baseline} introduces the baseline model and 
Homo Moralis preferences. Section~\ref{sec:static} 
characterizes equilibrium under complete information, deriving the 
moral Laffer curve and the elite's optimal allocation. 
Section~\ref{sec:dynamics-aligned} extends the framework to 
asymmetric information, characterizes the signaling equilibria 
that produce the jump-start dynamic, and discusses robustness. 
Section~\ref{sec:discussion} discusses the relation to 
reciprocity-based models, the behavioral foundations for 
moral universalization, and historical illustrations. 
Section~\ref{sec:conclusion} concludes.

\section{Baseline Model}\label{sec:baseline}
We consider an economy with two types of agents: a ruling elite that sets fiscal policy, and a unit mass of identical tax-paying citizens. The elite acts as a collective decision-maker, while citizens take policy as given and decide how much to comply. All formal proofs are collected in the Mathematical Appendix.

Each citizen earns the same income $w > 0$. Citizens choose an income report $\tilde{w} \geq 0$ to the tax authority. Truthful reporting corresponds to $\tilde{w} = w$, while underreporting corresponds to $\tilde{w} < w$. Because citizens are ex-ante identical and atomistic, the equilibrium report is common across citizens (see the discussion following Lemma~\ref{lem:optimal_report}). We omit the citizen index throughout and write $\tilde{w}$ for this common report.

The elite sets fiscal policy to maximize its own payoff subject 
to institutional constraints, while compliance behavior is 
shaped by moral preferences introduced in the next subsection. 
The interaction takes the form of a two-stage game, formally 
defined in Section~\ref{sec:static}: the elite moves first, 
choosing fiscal policy. Citizens then observe the policy and 
choose income reports. In the dynamic extension 
(Section~\ref{sec:dynamics-aligned}), the game is embedded in 
an infinite-horizon environment with asymmetric information.

The material payoff of each citizen depends on both public and private consumption and is assumed linear:
\begin{equation}
\pi(G, y) = \alpha \cdot G + y,
\end{equation}
where $G$ is a public good financed by tax revenue, $y$ is private consumption, 
and $\alpha>0$ captures the marginal value of the public good. 
For the baseline analysis we assume $\alpha$ is common to citizens and the elite. \ref{app:unaligned} extends the analysis to the unaligned case, in which the elite and citizens have distinct valuations $\alpha_E$ and $\alpha_C$, respectively.

\subsection{Policy and institutions}

The elite sets fiscal policy, which consists of four elements: the income tax rate $t$, public good provision $G$, transfers to the elite $B$, and transfers to citizens $b$. Institutional strength is captured by a parameter $\sigma \in (0,1)$, which governs the extent to which elite appropriation must be matched by transfers to citizens. For every unit appropriated by the elite, institutions require $\sigma$ units to be transferred to citizens, so that $b=\sigma B$. The elite's effective share of residual revenue is $\theta(\sigma)=1/(1+\sigma)$, strictly decreasing in $\sigma$, with $\theta(0)=1$ and $\theta(1)=1/2$.

This specification follows \citet{besley2011pillars}. The parameter $\sigma$ summarizes the accountability institutions that constrain elite appropriation. At $\sigma\to 0$ diverted revenue flows entirely to the elite. At $\sigma\to 1$ each unit diverted is matched one-for-one by a citizen transfer, leaving the elite with half. The composite $\sigma\,\theta(\sigma) = \sigma/(1+\sigma)$, the fraction of diverted revenue recaptured by citizens as transfers, will recur throughout the analysis.

For analytical convenience, we parameterize the allocation of tax revenue by a share $g\in[0,1]$, so that $G = gT$ and $B = \theta(\sigma)(1-g)T$, where $T$ denotes aggregate tax revenue. The elite's problem then reduces to a choice over $(t,g)$.

\subsection{Compliance under universalization}

Citizens choose a reported income $\tilde{w}\ge 0$ to the tax authority. Post-tax, pre-transfer income is
\begin{equation}
z(\tilde{w}) = w - t\,\tilde{w} - c\, C(\tilde{w}-w),
\end{equation}
where $c>0$ measures enforcement intensity and $C:\mathbb{R}\to\mathbb{R}_+$ is a convex misreporting-cost function of the deviation $d=\tilde{w}-w$, with $C(0)=0$, $C'(0)=0$, and $C''(d)>0$.\footnote{This reduced-form cost admits two interpretations. Under a \emph{material} reading, $cC(\tilde w - w)$ is the expected penalty from detection, and $c$ indexes the government's detection effort, in parallel with the strategic investments in coercive capacity studied by \citet{besley2009origins}, \citet{besley2011pillars}, and \citet{besley2020}. Under a \emph{lying-cost} reading, it is an intrinsic psychic cost of dishonesty, and $c$ scales the citizen's aversion to untruthful reporting. Because the term enters the citizen's actual and universalized net-income expressions symmetrically, the two readings yield identical first-order conditions and equilibrium objects. We adopt the material reading as the default.} Final private consumption is
\begin{equation}
y(\tilde{w}) = b + z(\tilde{w}).
\end{equation}

Throughout the analysis, we adopt the quadratic specification 
$C(d)=\tfrac{1}{2}d^2$, which delivers closed-form solutions 
while preserving the key convexity properties.\footnote{This 
specification penalizes deviations from truthful reporting 
symmetrically. Since Assumption~\ref{ass:kappa} below ensures 
that the optimal report never exceeds true income, the cost of 
over-reporting is never incurred in equilibrium and the 
symmetric form is without loss.}

\paragraph*{Universalization benchmark} If all citizens report the same $\tilde{w}$, per-capita objects are
\begin{align*}
T^{\mathcal{M}}(\tilde{w}) &= t\,\tilde{w} && \text{(Tax revenue)} \\
G^{\mathcal{M}}(\tilde{w}) &= g\, T^{\mathcal{M}}(\tilde{w}) && \text{(Public good)} \\
b^{\mathcal{M}}(\tilde{w}) &= \sigma\,\theta(\sigma)\,(1-g)\, T^{\mathcal{M}}(\tilde{w}) && \text{(Transfers to citizens)} \\
z^{\mathcal{M}}(\tilde{w}) &= w - t\,\tilde{w} - c\, C(\tilde{w}-w) && \text{(Net income)}.
\end{align*}

\begin{definition}[Homo Moralis utility]\label{def:HM}
Each citizen is characterized by a degree of morality $\kappa \in [0,1]$.\footnote{We follow the convention of \citet{alger2013homo}: $\kappa=0$ is the Homo Economicus benchmark and $\kappa=1$ is the pure-Kantian limit. The regularity bound $\kappa\alpha < 1-c/2$ in Assumption~\ref{ass:kappa} below ensures that the Laffer-maximizing tax rate is interior and revenue remains finite.}
Given a report $\tilde{w}$, the citizen’s utility is
\begin{equation}\label{eq:HM}
U^{(\kappa)}(\tilde{w}) 
= (1-\kappa)\,\pi\!\left(G,\, b+z(\tilde{w})\right) 
+ \kappa\,\pi\!\left(G^{\mathcal{M}}(\tilde{w}),\, b^{\mathcal{M}}(\tilde{w})+z^{\mathcal{M}}(\tilde{w})\right),
\end{equation}
where $\pi(G,y)=\alpha G + y$ is the material payoff function.
\end{definition}

The Homo Moralis utility is a convex combination of two payoff evaluations. The citizen attaches weight $1-\kappa$ to a selfish evaluation, in which policy is taken as given, and weight $\kappa$ to a universalized evaluation, in which she assesses the public good provision and transfers that would obtain if everyone submitted the same report\footnote{This counterfactual ties individual compliance decisions to 
the macroeconomic constraints of the economy, in the sense of 
\citet{laffont1975macroeconomic}.}.

Given the macroeconomic constraints and the linear tax structure, the fiscal variables (public good provision $G$, transfers to citizens $b$, and elite rents $B$) can be written as functions of the tax rate $t$, the public good share $g$, and citizens’ reports. For any given policy parameters $(t,g)$, each citizen chooses a report $\tilde{w}\ge 0$ to solve
\begin{equation*}
\begin{split}
\tilde{w}^*(\kappa; g,\alpha,t,c,\sigma)
= \arg\max_{\tilde{w}\ge 0}\;&
(1-\kappa)\,\pi\!\left(G,\, b+z(\tilde{w})\right) \\
&+ \kappa\,\pi\!\left(G^{\mathcal{M}}(\tilde{w}),\, 
   b^{\mathcal{M}}(\tilde{w})+z^{\mathcal{M}}(\tilde{w})\right),
\end{split}
\end{equation*}
where $\pi(G,y)=\alpha G+y$ is the material payoff function and 
\begin{equation}
z(\tilde{w}) = w - t\,\tilde{w} - c\,C(\tilde{w}-w)
\end{equation}
denotes post-tax, pre-transfer income.

The following lemma characterizes the citizen's optimal report 
under the quadratic specification.
\begin{lemma}[Optimal report]\label{lem:optimal_report}
Under the quadratic specification $C(d)=\tfrac{1}{2}d^2$, any interior solution to the citizen’s problem satisfies
\begin{equation}
\tilde{w}^*(\kappa; g,\alpha,t,c,\sigma)
= w + \frac{t}{c}\,\Big[\kappa\,\varphi(g,\alpha,\sigma) - 1\Big],
\end{equation}
where
\[
\varphi(g,\alpha,\sigma) := g\,\alpha + (1-g)\,\sigma\,\theta(\sigma)
\]
is the effective moral return to reporting, a weighted average of the marginal value of public goods and transfers under universalized behavior.
\end{lemma}

Since each citizen is atomistic and ex ante identical, and $\pi$ is linear in $G$ and $b$, Lemma~\ref{lem:optimal_report} yields a unique best reply common to all citizens. As a result, the Stage~2 equilibrium is symmetric.\footnote{Since the optimal report in Lemma~\ref{lem:optimal_report} is linear in $\kappa$, aggregate compliance in a heterogeneous population depends only on the mean degree of morality. The representative-agent specification is therefore without loss of generality for aggregate fiscal capacity. \ref{app:cultural} develops the two-type case with endogenous shares.}

Compliance behavior is governed by the product 
$\kappa\,\varphi(g,\alpha,\sigma)$. The effective moral return 
$\varphi$ aggregates the value of the public good, the 
allocation share to provision, and the degree of institutional 
matching, so that moral agents respond endogenously to the 
government's expenditure mix. If 
$\kappa\,\varphi(g,\alpha,\sigma)<1$, citizens under-report. If 
$\kappa\,\varphi(g,\alpha,\sigma)=1$, they report truthfully. 
If $\kappa\,\varphi(g,\alpha,\sigma)>1$, they over-comply. The 
benchmark $\kappa=0$ reproduces the purely selfish case in 
which citizens always under-report, with the extent of 
concealment determined only by enforcement $c$ and the tax
rate $t$.

\rev{In what follows we normalize income to $w=1$ without loss of generality. We discuss this point when deriving the elite's allocation rule (Proposition~\ref{prop:Elite}).}

\begin{assumption}\label{ass:kappa}
The degree of morality satisfies $\kappa\,\alpha < 1 - c/2$.
\end{assumption}

This restriction does two things. First, it rules out 
over-compliance: whenever $\kappa\alpha \ge 1$, the universalized 
return to contribution under $g=1$ weakly dominates unity, so 
reported income weakly exceeds true income and feasible revenue 
is unbounded. Second, and strictly stronger, it ensures that the 
Laffer problem admits an interior solution: the revenue-maximizing 
tax rate $\hat{t}(g,\kappa,c,\sigma)$ defined in~\eqref{eq:that} 
satisfies $\hat{t}<1$ for every feasible $g$, so the mechanism 
operates in the regime where citizens respond to tax-rate changes 
through compliance adjustments rather than corner outcomes. All qualitative results are continuous in 
$\kappa$ and hold for any $\kappa$ bounded away from the boundary 
$\kappa\alpha = 1-c/2$.

\section{The static framework}\label{sec:static}

This section characterizes equilibrium under complete 
information about the value of public goods. We first define 
the game formally, then solve for citizens' compliance and the 
resulting fiscal capacity, and finally characterize the elite's 
optimal allocation.

\subsection{Game form and equilibrium concept}\label{sec:game-form}

\begin{definition}[Static fiscal game]\label{def:static-game}
The \emph{static fiscal game} $\Gamma(\alpha,\kappa,c,\sigma)$ is defined by:
\begin{enumerate}
\item \textbf{Players.} A ruling elite~$E$, formalized as a single decision-maker, and a unit mass of identical citizens $i\in[0,1]$.
\item \textbf{Timing.}
  \begin{enumerate}
    \item[\emph{Stage~1.}] The elite publicly chooses a fiscal policy $(t,g)\in[0,1]^2$, where $t$ is the tax rate and $g$ is the share of revenue allocated to the public good. Transfers $B$ and $b$ are determined residually by the budget constraint and the institutional parameter $\sigma$.
    \item[\emph{Stage~2.}] Each citizen $i$ observes $(t,g)$, and all citizens simultaneously choose reports $\tilde{w}_i\geq 0$.
  \end{enumerate}
\item \textbf{Payoffs.}
  \begin{itemize}
    \item Each citizen's payoff is the Homo Moralis utility $U^{(\kappa)}(\tilde{w}_i)$ defined in equation~\eqref{eq:HM}.
    \item The elite's payoff is $U_E = \alpha\,G + B$, where $G = g\,T$ and $B = \theta(\sigma)(1-g)\,T$ are determined by aggregate revenue $T = \int_0^1 t\,\tilde{w}_i\,di$.
  \end{itemize}
\item \textbf{Information.} All parameters $(\alpha,\kappa,c,\sigma)$ are common knowledge.
\end{enumerate}
\end{definition}

As noted after Lemma~\ref{lem:optimal_report}, each citizen's best reply is independent of other citizens' reports, so the Stage~2 equilibrium is symmetric and coincides with the individual optimum characterized there.
The solution concept is \emph{Subgame Perfect Equilibrium} 
(SPE): in Stage~2, each citizen plays the unique best response 
$\tilde{w}^*(\kappa;g,\alpha,t,c,\sigma)$ from 
Lemma~\ref{lem:optimal_report}. In Stage~1, the elite 
anticipates this response and chooses $(t^*,g^*)$ to maximize 
$U_E$. Figure~\ref{fig:timeline-static} 
summarizes the timing.

\begin{figure}[!htb]
\centering
\includegraphics{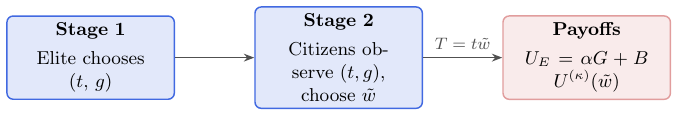}
\caption{Timing of the static fiscal game $\Gamma(\alpha,\kappa,c,\sigma)$. 
The elite moves first, choosing fiscal policy $(t,g)$. Citizens then observe policy and choose income reports $\tilde w$. Payoffs are realized after revenue $T=t\tilde w$ is generated. All parameters $(\alpha,\kappa,c,\sigma)$ are common knowledge, and the game is solved by backward induction.}
\label{fig:timeline-static}
\end{figure}

\subsection{Fiscal capacity and the Laffer curve}

Fiscal capacity is defined as the maximum tax revenue the government can raise, given citizens’ degree of morality $\kappa$ and the enforcement parameter $c$. 
Per-capita tax revenue, given a tax rate $t$ and the reporting choice $\tilde{w}(\kappa; g, \alpha, t, c, \sigma)$, is
\begin{equation}\label{eq:T}
\begin{split}
T(t, g, \kappa, c, \sigma)
&= t\,\tilde{w}(\kappa; g, \alpha, t, c, \sigma) \\
% R2 (w=1). Original coefficient: \frac{t\,w}{c}
&= \rev{\frac{t}{c}} \Big[c - t \big(1 - \kappa\, \varphi(g, \alpha, \sigma) \big) \Big].
\end{split}
\end{equation}

This expression yields a Laffer-type curve: tax revenue is hump-shaped in the tax rate, rising approximately linearly at low $t$ and eventually declining as reduced compliance offsets the mechanical rate increase.

The revenue-maximizing tax rate solves 
$\max_t T(t, g, \kappa, c, \sigma)$, yielding
\begin{equation}\label{eq:that}
\hat{t}(g,\kappa,c,\sigma) = \frac{c/2}{1 - \kappa\,\varphi(g,\alpha,\sigma)}.
\end{equation}
Substituting into \eqref{eq:T} gives the peak of the Laffer curve:
\begin{equation}\label{eq:Tpeak}
% R2 (w=1). Original: \frac{w c}{4\,[1 - \kappa\,\varphi]}
T(\hat{t}, g, \kappa, c, \sigma) = \rev{\frac{c}{4\,[1 - \kappa\,\varphi(g,\alpha,\sigma)]}}.
\end{equation}

Figure~\ref{fig:laffer} plots these curves for several values of $\kappa$ and both allocation policies $g \in \{0,1\}$.

\begin{figure}[!htb]
    \centering
    \includegraphics[width=0.75\linewidth]{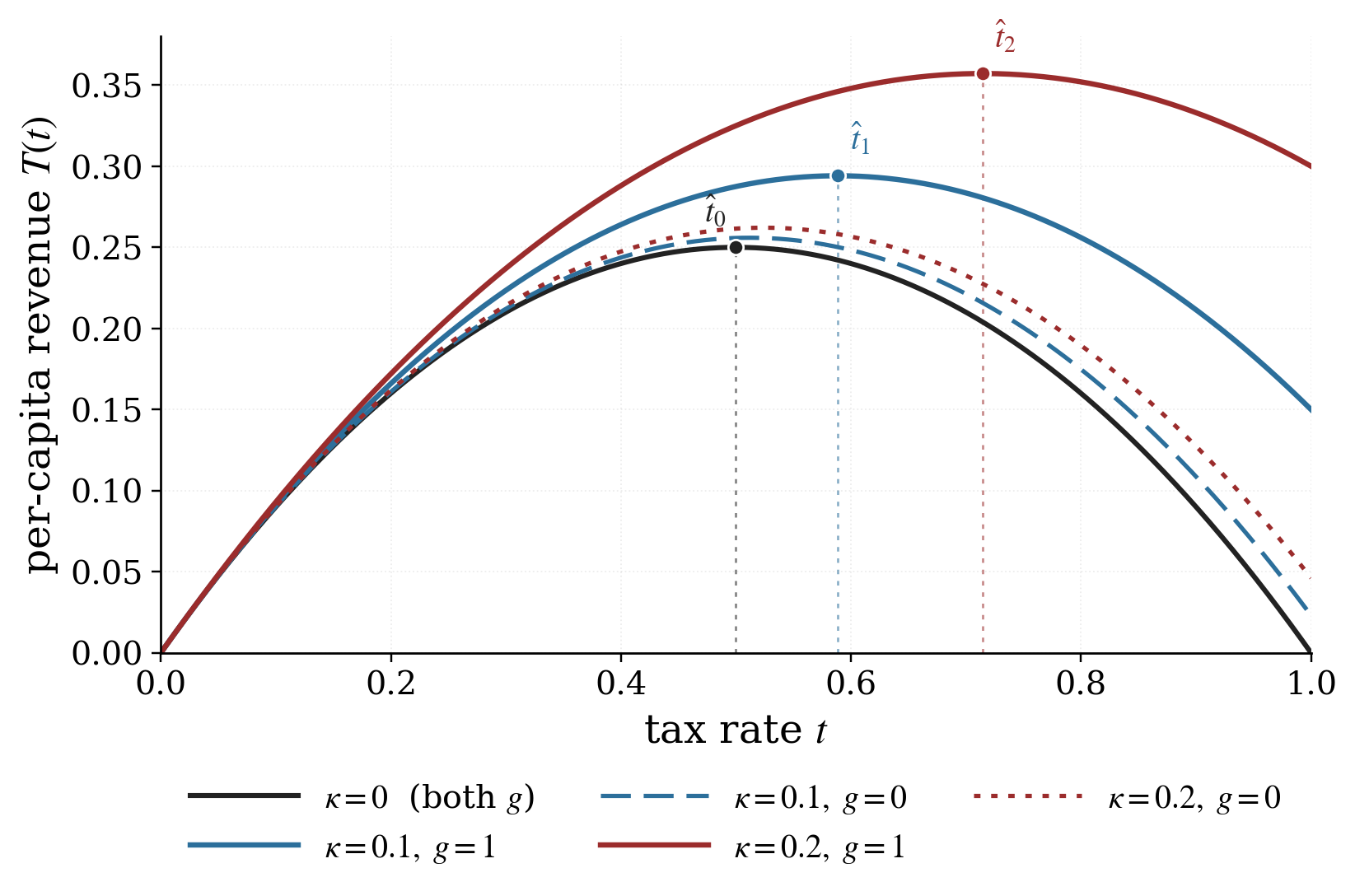}
    \caption{Moral Laffer curves. Per-capita tax revenue $T(t)$ as a 
    function of the tax rate for three morality levels 
    $\kappa\in\{0,0.1,0.2\}$ and two allocation policies 
    $g\in\{0,1\}$, evaluated at $\alpha=1.5$, $\sigma=0.3$, and 
    $w=c=1$. At $\kappa=0$ the two allocation policies yield 
    identical revenue, shown as a single solid black line. For 
    $\kappa>0$, solid lines denote full provision ($g=1$) and 
    dashed lines denote full rents ($g=0$). Filled dots mark the 
    revenue-maximizing tax rate $\hat{t}$ on each $g=1$ curve, 
    with vertical dashed lines indicating its position on the 
    tax-rate axis.}
    \label{fig:laffer}
\end{figure}

Three features of Figure~\ref{fig:laffer} bear on the mechanism 
developed below. First, at $\kappa=0$ compliance depends only on 
the tax rate and enforcement, so the allocation of revenue is 
irrelevant: the Laffer curve is fixed. Second, as $\kappa$ rises 
the whole curve shifts up and the revenue-maximizing rate 
$\hat{t}$ shifts to the right, jointly enlarging the height and 
location of fiscal capacity. Third, and most important for what 
follows, the gap between the $g=1$ (solid) and $g=0$ (dashed) 
curves widens with $\kappa$: this is the \emph{moral premium for 
provision}. In the figure, the $g=0$ curves sit close to the 
$\kappa=0$ baseline because at these parameters the moral return 
under rents ($\kappa\sigma\theta(\sigma)$) is several times 
smaller than the moral return under provision ($\kappa\alpha$). 
Because $\varphi(1,\alpha,\sigma)=\alpha$ substantially 
exceeds $\varphi(0,\alpha,\sigma)=\sigma\theta(\sigma)$ when 
institutions are weak, morality couples fiscal capacity to the 
allocation decision. This coupling is what generates the elite 
discipline characterized in Proposition~\ref{prop:Elite}: the tax 
base the elite can extract is not a primitive of the economy but a 
function of how it chooses to use revenue.

\subsection{The elite’s problem}

The elite’s utility is given by $U_E = \alpha G + B$, where $G$ denotes public good provision and $B$ rents.  
The elite chooses $(t, g)$ to maximize this utility.  
Conveniently, the problem can be re-expressed as a choice over $g \in [0,1]$, the share of tax revenue allocated to the public good.  

For each $g$, the tax rate $\hat{t}(g,\kappa,c,\sigma)$ is chosen to maximize revenue, and the resulting allocations satisfy
\begin{equation}\label{eq:G-B-constraints}
G(g) = g \cdot \hat{T}(g),
\qquad  
B(g) = \theta(\sigma)\,(1-g)\,\hat{T}(g),
\end{equation}
where
\begin{equation}
\hat{T}(g) := T(\hat{t}(g,\kappa,c,\sigma), g, \kappa, c, \sigma)
\end{equation}
denotes per-capita revenue under the Laffer-maximizing tax rate at allocation share $g$.  

Using the constraints in equation~\eqref{eq:G-B-constraints}, 
the elite's problem reduces to
$\max_{g \in [0,1]} U_E(\alpha; g)$, where
\begin{equation}\label{eq:Elite}
U_E(\alpha; g) = \hat{T}(g)\,[\alpha g + \theta(\sigma)(1-g)].
\end{equation}

Because the tax rate is optimized for each $g$, the problem 
reduces to a one-dimensional maximization over the allocation 
share $g$. The elite's objective $U_E(\alpha;g)$ is the product 
of two functions that are affine in $g$, so it is single-peaked 
and the optimum is always at a corner. The following result 
characterizes which corner obtains.\footnote{\rev{Having defined the elite's objective, we elaborate on the normalization $w=1$. For general $w$, the misreporting cost $C(d)=d^{2}/2$ depends on the deviation $d=\tilde{w}-w$ alone, so equation~\eqref{eq:T} reads $T=(t/c)\,[\,wc-t(1-\kappa\varphi)\,]$, the revenue-maximizing rate is $\hat t=wc/[2(1-\kappa\varphi)]$, and the peak of the Laffer curve~\eqref{eq:Tpeak} is $\hat T(g)=w^{2}c/[4(1-\kappa\varphi(g,\alpha,\sigma))]$. Since revenue enters the elite's objective~\eqref{eq:Elite} multiplicatively, general $w$ scales that objective by the positive constant $w^{2}$. The normalization thus does not affect the elite's optimal strategy or the thresholds characterized below. A more general treatment would require Assumption~\ref{ass:kappa} to read $\kappa\alpha<1-wc/2$.}\label{fn:wnorm}}

\begin{proposition}[Elite's allocation under alignment]\label{prop:Elite}
The elite's optimal allocation $g^*\in\{0,1\}$ satisfies:

\begin{enumerate}
    \item \textbf{Strong provision state.}  
    If $\alpha > \theta(\sigma)$, provision is always optimal, i.e. $g^* = 1$ for all feasible $\kappa$.

    \item \textbf{Weak provision state.}  
    If $\sigma\,\theta(\sigma) < \alpha \le \theta(\sigma)$, there exists a unique morality threshold
    \begin{equation}
    \bar\kappa(\alpha,\sigma)
    = \frac{\theta(\sigma)-\alpha}{\alpha\,\theta(\sigma)\,(1-\sigma)}
    \in \big(0,\,1/\alpha\big),
    \end{equation}
    such that
    \begin{equation}
g^* =
    \begin{cases}
    0, & \text{if } \kappa < \bar\kappa(\alpha,\sigma),\\[0.4em]
    1, & \text{if } \kappa \ge \bar\kappa(\alpha,\sigma).
    \end{cases}
    \end{equation}

    \item \textbf{Transfer state.}  
    If $\alpha \le \sigma\,\theta(\sigma)$, the elite always prefers rents, i.e. $g^* = 0$ for all feasible $\kappa$.
\end{enumerate}

Moreover, the threshold $\bar\kappa(\alpha,\sigma)$ satisfies:
\begin{equation}
\frac{\partial \bar\kappa}{\partial \alpha}<0,
\qquad
\frac{\partial \bar\kappa}{\partial \sigma}
\begin{cases}
<0 & \text{if } \alpha>\tfrac{1}{2},\\
=0 & \text{if } \alpha=\tfrac{1}{2},\\
>0 & \text{if } \alpha<\tfrac{1}{2}.
\end{cases}
\end{equation}
Hence higher $\alpha$ always reduces the morality required to induce provision, and stronger institutions reduce (resp. increase) the required morality when $\alpha>\tfrac{1}{2}$ (resp. $\alpha<\tfrac{1}{2}$).

%\begin{proof}
%See \ref{app:proof-elite-alignment}.
%\end{proof}
\end{proposition}

The elite's allocation determines the equilibrium tax base, 
which responds to both the spending mix and citizens' morality. 
The three regions in Proposition~\ref{prop:Elite} admit a clear 
economic interpretation. In the \emph{strong provision state} 
($\alpha>\theta(\sigma)$), the elite's valuation of the public 
good exceeds the residual rents from diversion, and provision 
is chosen unconditionally. Morality still affects the magnitude
of the tax base but is inessential for the allocation itself. 
The \emph{transfer state} ($\alpha\le\sigma\theta(\sigma)$) is 
the mirror case: the public good is valued so weakly that no 
admissible morality can make provision optimal, and the elite 
diverts throughout. The paper's mechanism operates in the 
intermediate \emph{weak provision state}, where the elite's 
private ranking favors rents but lies close enough to 
indifference that morality can flip the allocation. Once 
citizens' internalization $\kappa$ crosses $\bar\kappa(\alpha,\sigma)$, 
universalized reasoning expands the tax base under $g=1$ enough 
to overturn the elite's preference for rents, and provision 
becomes optimal.

\begin{corollary}[Taxes raised in equilibrium]\label{cor:aligned-taxbase}
At the elite's optimum, the equilibrium tax base is given by
\begin{align}
    % R2 (w=1). Original numerators: \frac{w c}{4\,(\cdots)}
    T_1(\kappa,\sigma) &= \rev{\frac{c}{4\,(1 - \kappa\,\alpha)}} && \text{if } g^* = 1, \label{eq:T1}\\
    T_0(\kappa,\sigma) &= \rev{\frac{c}{4\,(1 - \kappa\,\sigma\,\theta(\sigma))}} && \text{if } g^* = 0. \label{eq:T0}
\end{align}
Denote by $T^*(\kappa,\sigma)$ the tax base at the elite's optimal allocation $g^*$.
\end{corollary}

Two features of the threshold $\bar\kappa(\alpha,\sigma)$ deserve 
emphasis. First, $\partial\bar\kappa/\partial\alpha<0$: morality 
and fundamentals are \emph{substitutes} in inducing provision. 
Societies with a low value of public spending can 
compensate through strong internalized norms, and vice versa. 
This substitutability is the basis of the moral channel identified 
in the introduction: the ability of weak states to sustain fiscal 
capacity through citizen universalization rather than coercive 
enforcement alone. Second, $\partial\bar\kappa/\partial\sigma$ 
changes sign at $\alpha=1/2$: morality and institutions are 
\emph{complements} when the public good is sufficiently valuable 
($\alpha>1/2$) and \emph{substitutes} otherwise. In the 
complementary regime, better institutions \emph{reduce} the 
morality required to induce provision: coercive capacity and 
moral capacity work in tandem to discipline the elite. In the 
substitutive regime, stronger institutions \emph{raise} the 
threshold: the two channels compete for doing the same work, and 
one partially crowds out the other.\footnote{This regime distinction is only relevant where the
threshold is interior. Under Assumption~\ref{ass:kappa}, that requires
$\alpha$ to exceed $\alpha_{\min}(c,\sigma)=\theta(\sigma)\,[1-(1-c/2)(1-\sigma)]$.
Notably, $\alpha_{\min}$ equals $1/2$ when $c=1$, for every $\sigma$,
exactly the value at which the complementary and substitutive regimes switch.}

Beyond the allocation margin, the equilibrium tax base 
$T^*(\kappa,\sigma)$ is monotone in $\kappa$: moral agents comply 
more, expanding capacity regardless of the allocation. The 
expansion is greatest under provision, since the denominator 
$1-\kappa\alpha$ at $g^*=1$ falls more steeply in 
$\kappa$ than $1-\kappa\sigma\theta(\sigma)$ does at $g^*=0$. Stronger morality thus operates on two margins 
simultaneously: it amplifies the fiscal returns to provision and 
shifts the elite's optimal allocation toward it, disciplining 
elites into providing public goods even when $\alpha<\theta(\sigma)$.

\section{Dynamics and asymmetric information}\label{sec:dynamics-aligned}

The static model takes the value of the public good as given. In practice, fiscal fundamentals shift over time, and the elite, through privileged access to administrative data, fiscal projections, and policy expertise, may be better informed than citizens about the current state. This section extends the framework to an infinite-horizon environment in which the elite privately observes the value of public spending each period, so that its allocation choice becomes a potential signal of fiscal fundamentals.

The central message is that moral universalization enables 
credible fiscal reform under asymmetric information, even when 
institutions are weak. High-value elites signal their type 
through provision, citizens infer the improvement in fundamentals 
and raise compliance, and the tax base expands within the same 
period. Lemma~\ref{lem:stage-reduction} shows that equilibrium 
behavior in this signaling environment admits a tractable 
stage-game characterization, and 
Propositions~\ref{prop:weak-high}--\ref{prop:strong-high} identify 
the morality thresholds that sustain credible reform together 
with the resulting fiscal multiplier.

\subsection{Environment and timing}

The value of the public good follows a two-state Markov chain privately observed by the elite at the start of each period.\footnote{In \citet{besley2020}, the value of public spending is publicly observed each period, so no signaling game arises. Here, the elite's private observation makes its allocation choice informative about fiscal fundamentals.} As in the static model, Assumption~\ref{ass:kappa} restricts morality to $\kappa\,\alpha^H < 1 - c/2$.

\begin{definition}[Dynamic fiscal signaling game]\label{def:dynamic-game}
The \emph{dynamic fiscal signaling game} extends $\Gamma$ to an 
infinite-horizon setting with asymmetric information and time-varying 
fundamentals. Primitives, technology, and action spaces are as in 
Definition~\ref{def:static-game}. Time is discrete, 
$n = 0, 1, 2, \ldots$. Each period proceeds in three stages.

\medskip\noindent\textbf{Stage 0.}\quad Nature draws the value of the public good 
$\alpha_n \in \{\alpha^{L},\alpha^{H}\}$ from a Markov chain 
with transition probabilities 
$q^{H}\equiv\Pr(\alpha_{n+1}=\alpha^{H}\mid \alpha_n=\alpha^{H})$ 
and 
$q^{L}\equiv\Pr(\alpha_{n+1}=\alpha^{H}\mid \alpha_n=\alpha^{L})$, 
with $q^{H}\ge q^{L}$. The chain is initialized at its stationary 
distribution, so the marginal probability is
\[
\Pr(\alpha_n=\alpha^{H}) \;=\; \rho \;\equiv\; \frac{q^{L}}{1-q^{H}+q^{L}} \;\in\; (0,1)
\]
each period. At each $n$ 
the elite privately observes 
$\alpha_n$.\footnote{The specification covers two limiting cases: 
$q^{H}=q^{L}$ yields i.i.d.\ draws on $\{\alpha^{L},\alpha^{H}\}$ 
(with $\rho=q^{L}$), the two-state analog of the i.i.d.\ 
benchmark in \citet{besley2020}, and $(q^{H},q^{L})=(1,0)$ 
yields an absorbing fixed-$\alpha$ environment, in which case 
$\rho$ is exogenously pinned down by the initial draw rather 
than by the stationary formula.}

\medskip\noindent\textbf{Stage 1.}\quad The elite, knowing $\alpha_n$, publicly 
chooses a fiscal policy $(t_n,g_n)\in[0,1]^2$.

\medskip\noindent\textbf{Stage 2.}\quad Citizens observe 
$(t_n,g_n)$, form a posterior $p_n$ over $\alpha_n$ by Bayes' rule 
(on path) and the D1 refinement (off path), and choose reports 
$\tilde{w}_n$ to maximize the {Homo Moralis} utility 
$U^{(\kappa)}$ in equation~\eqref{eq:HM}, with universalization 
applied to the stage-game report rather than to a 
history-contingent strategy.\footnote{This is a modeling 
choice. It is consistent with an overlapping-generations 
interpretation in which each cohort of citizens inherits 
$\kappa$ from its predecessors but has no stake in prior 
periods. Universalizing over strategies would admit 
history-contingent plans, including grim-trigger punishments 
of low-type elites, as equilibrium objects.  Beliefs about the 
fundamental $\alpha$ remain unrestricted: citizens may be 
arbitrarily forward-looking about future states.}

\medskip\noindent\textbf{Strategies.}\quad
Strategies are Markov. The elite's strategy is a mapping $\gamma_E\colon \{\alpha^{L},\alpha^{H}\}\to [0,1]^2$ from the current state to a fiscal policy $(t_n,g_n)$. Each citizen's strategy is a mapping $\gamma_C\colon [0,1]^2\times[0,1]\to\mathbb{R}_+$ from the observed policy and posterior $p_n$ to a report $\tilde{w}_n$.

\medskip\noindent\textbf{Payoffs.}\quad
Each citizen's per-period payoff is $U^{(\kappa)}(\tilde{w}_n)$ 
as specified in Stage~2. The elite's per-period payoff is
\begin{equation}\label{eq:UEdynamic}
U_E(\alpha_n;\, g_n \mid p_n) 
\;=\; T(g_n \mid p_n)\,
\bigl[\alpha_n\, g_n + \theta(\sigma)\,(1-g_n)\bigr],
\end{equation}
where $T(g_n \mid p_n) \equiv g_n\,T_1(p_n) + (1-g_n)\,T_0$ is the 
belief-dependent tax base, with
\begin{equation}\label{eq:T-belief}
% R2 (w=1). Original numerators: \frac{w c}{4\,(\cdots)}
T_1(p_n) \;\equiv\; \rev{\frac{c}{4\,(1-\kappa\,p_n)}},
\qquad
T_0 \;\equiv\; \rev{\frac{c}{4\,(1-\kappa\,s)}},
\qquad
s \;\equiv\; \sigma\,\theta(\sigma).
\end{equation}
The elite maximizes the discounted payoff stream 
$\sum_{n=0}^{\infty}\delta^{n}\,U_E(\alpha_n;\,g_n\mid p_n)$ 
for discount factor $\delta\in[0,1)$. 
Lemma~\ref{lem:stage-reduction} shows that the stage-game 
payoff comparison is sufficient for the elite's incentive 
constraint in every equilibrium configuration.

\medskip
All parameters $(\kappa,c,\sigma,q^{H},q^{L})$ are common 
knowledge. Figure~\ref{fig:timeline-dynamic} summarizes the 
within-period timing.
\end{definition}

\begin{figure}[!htb]
\centering
\includegraphics{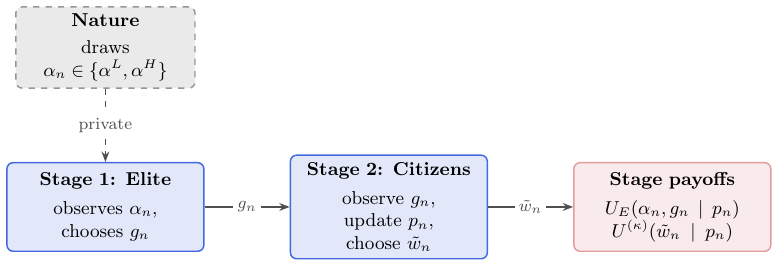}
\caption{Timing of the dynamic fiscal signaling game
(Definition~\ref{def:dynamic-game}). At the start of period $n$,
Nature draws $\alpha_n\in\{\alpha^L,\alpha^H\}$ from a Markov chain with
transition probabilities $(q^H, q^L)$. The dashed box and arrow indicate
that $\alpha_n$ is privately observed by the elite. The elite then chooses
$g_n$ publicly. Citizens observe $g_n$, update their belief $p_n$ via
Bayes' rule (on path) or the D1 refinement \citep{banks1987equilibrium}
(off path), and choose $\tilde{w}_n$.}
\label{fig:timeline-dynamic}
\end{figure}
For any fixed belief $p$, the per-period payoff $U_E$ is 
convex in $g$ (it is the product of two affine functions of 
$g$). Whether this stage-game property extends to the dynamic 
optimization depends on whether the continuation value favors 
interior allocations. Lemma~\ref{lem:stage-reduction} below 
shows it does not. Under provision, the belief-dependent tax 
base reduces to $T_1(p_n)$, generalizing the static provision 
tax base (equation~\eqref{eq:T1}) by replacing the known 
$\alpha$ with the posterior belief~$p$. Under rents it 
reduces to $T_0$, coinciding with equation~\eqref{eq:T0}.

\paragraph*{Solution concept} The equilibrium concept is Markov Perfect Equilibrium (MPE) \citep{maskin2001markov}, refined by the D1 
criterion\footnote{Informally, D1 disciplines 
beliefs about unexpected actions by asking which type of elite 
would have the strongest incentive to deviate. If the high-value 
elite benefits from an unexpected deviation for a strictly larger 
set of citizen responses than the low-value elite does, D1 
requires citizens to attribute the deviation to the high-value elite. 
In our setting, unexpected provision is therefore interpreted as 
evidence of high $\alpha$. \rev{Strictly speaking, the dynamic environment studied here is not a canonical sender--receiver signaling game, so our use of D1 is an adaptation. It remains appropriate because the within-period interaction retains signaling structure, with a privately informed elite taking an action that citizens observe before responding, and because the type set is fixed (the state $\alpha$), so the type-by-type belief comparisons on which D1 rests are well-defined. This is a milder departure than in settings where the type set is endogenous and such comparisons break down. \citet{ekmekci2023signaling} discuss the difficulty of applying standard belief-based refinements outside the canonical signaling game.}} \citep{banks1987equilibrium}.
Equilibrium requires \emph{belief consistency} (citizens' 
posteriors follow the elite's equilibrium strategy via Bayes' 
rule on path and the D1 criterion off path) and \emph{incentive 
compatibility} (given these beliefs, each elite type prefers 
its equilibrium action to imitating the other). Together, these 
conditions ensure that observed reforms are credible and 
self-enforcing.

\paragraph*{Stage-game payoff gain from provision}

The elite’s incentive to provide rather than extract is captured by the stage-game payoff difference:
\begin{equation}\label{eq:DeltaAligned}
\begin{split}
\Delta(\alpha_n \mid p_n)
&= U_E(\alpha_n; g_n{=}1 \mid p_n) - U_E(\alpha_n; g_n{=}0 \mid p_n) \\
% R2 (w=1). Original prefactor: \frac{w c}{4}
&= \rev{\frac{c}{4}}
\left[
\frac{\alpha_n}{1 - \kappa\, p_n}
-
\frac{\theta(\sigma)}{1 - \kappa\, \sigma\,\theta(\sigma)}
\right],
\end{split}
\end{equation}
where $p_n$ denotes citizens’ belief about $\alpha_n$. For any given belief $p_n$, the sign of $\Delta(\alpha_n \mid p_n)$ determines whether provision is optimal for a type with value $\alpha_n$. 
The function $\Delta(\alpha_n \mid p_n)$ is strictly increasing in $\alpha_n$, implying a \emph{single-crossing property}: 
elite best replies can therefore be represented by a cutoff rule in~$\alpha_n$.

Because citizens' beliefs carry forward across periods, the 
elite's period-$n$ choice affects future compliance as well as 
the current payoff. The next lemma shows that this continuation 
channel does not overturn the within-period incentive, so 
equilibrium behavior can be characterized by the stage-game 
payoff gain alone.

\begin{lemma}[Stage-game reduction]\label{lem:stage-reduction}
In any MPE of the game in 
Definition~\ref{def:dynamic-game}, the elite's period-$n$ 
action depends only on the current state $\alpha_n$ and the 
contemporaneous citizen belief $p_n$. The stage-game 
condition
\[
\Delta\bigl(\alpha_n \,\big|\, p_n(g_n{=}1)\bigr) \ \ge\ 0
\]
is sufficient for provision to be a best response, and exact 
whenever future beliefs depend on the period-$n$ action only 
through the realized state $\alpha_n$ (see the characterization in case 3 of Proposition~\ref{prop:weak-high}).
\end{lemma}
The proof (\ref{app:stage-reduction}) applies the one-shot 
deviation principle to each equilibrium configuration, showing 
that a deviation in $g_n$ either leaves the continuation value 
unchanged or weakly lowers it.
Combined with the convexity of $U_E$ in $g_n$, the stage-game 
reduction implies that the corner result $g^*_n\in\{0,1\}$ 
established in Section~\ref{sec:static} extends to the dynamic game.

\subsection{Equilibrium characterization}\label{sec:eq-characterization}

By Lemma~\ref{lem:stage-reduction}, equilibrium conditions 
reduce to stage-game incentive constraints. We seek a separating 
equilibrium in which $g^{*}(\alpha^{H})=1$ and 
$g^{*}(\alpha^{L})=0$. Under separation, on-path beliefs are 
correct: citizens infer $\alpha^{H}$ upon observing $g=1$ and 
$\alpha^{L}$ upon observing $g=0$, so compliance expands from 
the rents tax base $T_{0}$ to the provision tax base 
$T_{1}(\alpha^{H})$ in equation~\eqref{eq:T-belief}.

The two incentive-compatibility conditions correspond to opposite
deviations and must be evaluated at the relevant beliefs.

\emph{High-type IC.} If the high-value elite deviates from $g=1$ to
$g=0$, citizens infer $\alpha^{L}$. Because the rents tax base $T_{0}$
in equation~\eqref{eq:T-belief} does not depend on beliefs about
$\alpha$, the deviation payoff coincides with the on-path comparison
at $p=\alpha^{H}$:
\begin{equation}\label{eq:IC-H}
% R2 (w=1). Original prefactor: \frac{w c}{4}
\Delta(\alpha^{H}\mid \alpha^{H}) \;=\;
\rev{\frac{c}{4}}\!\left[\frac{\alpha^{H}}{1-\kappa\,\alpha^{H}}
- \frac{\theta(\sigma)}{1-\kappa\, s}\right] \;\ge\; 0,
\end{equation}
where $s\equiv\sigma\,\theta(\sigma)$. This yields the lower threshold
\begin{equation}\label{eq:kappa-Hmin-text}
\kappa \;\ge\; \kappa^{H}_{\min}
\;\equiv\; \frac{\theta(\sigma)-\alpha^{H}}{\theta(\sigma)\,\alpha^{H}\,(1-\sigma)}.
\end{equation}

\emph{Low-type IC.} If the low-value elite deviates from $g=0$ to 
$g=1$, citizens observe an action that, in the separating candidate, 
is played on path by the high type with probability one. Bayes' rule 
therefore delivers $p(g=1)=\alpha^{H}$, and the low type's deviation 
payoff is $T_{1}(\alpha^{H})\,\alpha^{L}$, evaluated at the 
high-type belief,\footnote{Evaluating the low-type IC at the prior 
$p=\alpha^{L}$ rather than at the Bayesian posterior $p=\alpha^{H}$ 
would yield a higher upper threshold and thus a wider separating 
interval. The Bayesian posterior is therefore the conservative 
benchmark: it imposes the strongest credibility requirement on the 
low type. The D1 refinement (\ref{app:D1}) plays a parallel 
role off path in the pool-at-rents candidate, where it likewise 
selects $p(g=1)=\alpha^{H}$ and thereby rules out the deviation.} 
and the IC for separation is
\begin{equation}\label{eq:IC-L}
% R2 (w=1). Original prefactor: \frac{w c}{4}
\Delta(\alpha^{L}\mid \alpha^{H}) \;=\;
\rev{\frac{c}{4}}\!\left[\frac{\alpha^{L}}{1-\kappa\,\alpha^{H}}
- \frac{\theta(\sigma)}{1-\kappa\, s}\right] \;\le\; 0,
\end{equation}
yielding the upper threshold
\begin{equation}\label{eq:kappa-Lmax-text}
\kappa \;\le\; \kappa^{L}_{\max}
\;\equiv\; \frac{\theta(\sigma)-\alpha^{L}}
{\theta(\sigma)\,\alpha^{H} - \alpha^{L}\, s}.
\end{equation}
Together, \eqref{eq:kappa-Hmin-text} and \eqref{eq:kappa-Lmax-text}
pin down the separating interval $[\kappa^{H}_{\min},\kappa^{L}_{\max})$
within which credible reform is incentive-compatible. The thresholds
depend only on the signaling-game primitives
$(\alpha^{L}, \alpha^{H}, \sigma, c)$ and not on the transition
probabilities $(q^{H}, q^{L})$.

When morality is low ($\kappa<\kappa^{H}_{\min}$), even high-value
elites prefer rents, and the equilibrium pools at $g=0$. When morality
is high ($\kappa\ge\kappa^{L}_{\max}$), low-value elites also find
provision attractive and mimic, so separation collapses and the
equilibrium pools at $g=1$. Proposition~\ref{prop:weak-high} formalizes
the three regions and states the additional feasibility condition for
pooling at provision.

\paragraph*{A fundamentals floor for credible reform}\label{lem:admissibility}
Not all values of $\alpha^{H}$ above $\sigma\theta(\sigma)$ 
permit credible reform. Define
\[
\underline{\alpha}(c,\sigma)
\;\equiv\;
\theta(\sigma)\!\left[\sigma + \tfrac{c}{2}(1-\sigma)\right].
\]
When $\alpha^{H}\le\underline{\alpha}$, the morality threshold 
$\kappa^{H}_{\min}$ exceeds the upper bound imposed by 
Assumption~\ref{ass:kappa}, so no admissible $\kappa$ satisfies 
the high-type IC and the unique equilibrium is pooling at rents. 
(The condition $\kappa^{H}_{\min}\le(1-c/2)/\alpha^{H}$ 
rearranges to $\alpha^{H}\ge\underline{\alpha}$.) The gap 
$\underline{\alpha}-\sigma\theta(\sigma) = 
\tfrac{c}{2}\,\theta(\sigma)(1-\sigma)$ narrows with stronger 
institutions (higher~$\sigma$) and lower compliance costs 
(lower~$c$), so the constraint binds least when the environment 
is otherwise most favorable to reform. For the remainder of 
this section, we assume $\alpha^{H}>\underline{\alpha}$, which 
is the economically interesting case: the regime in which 
reform is not materially self-enforcing but morality can 
restore it.

\subsection{Weak high state: morality-activated reform}

Consider first the case in which the value of the public good in the high state is only moderately greater than in the low state, but sufficiently above the transfer threshold that morality can activate reform.
Formally, assume
\begin{equation}
\alpha^L < \sigma\,\theta(\sigma) < \underline{\alpha}(c,\sigma) < \alpha^H \le \theta(\sigma),
\end{equation}
so that even in the high state, provision is socially desirable for citizens but not strictly profitable for the elite. 
This configuration therefore represents a situation in which fundamentals are improving yet still too weak to make reform self-enforcing on material grounds. 
We refer to it as the \emph{weak-high state} because, although the value of the public good is higher than in the low state, the elite would not provide it voluntarily without moral internalization. 
In this region, morality is essential for credible reform: the elite provides only if the degree of morality is high enough to make the signal self-enforcing.

\begin{proposition}[Morality-activated reform and equilibrium characterization]\label{prop:weak-high}
Suppose $\alpha^{L} < \sigma\,\theta(\sigma) < \underline{\alpha} < \alpha^{H} \le \theta(\sigma)$, and define the morality thresholds
\begin{equation}\label{eq:kappa-min-max}
\kappa_{\min}^{H}
= \frac{\theta(\sigma)-\alpha^{H}}{\theta(\sigma)\,\alpha^{H}\,(1-\sigma)},
\qquad
\kappa_{\max}^{L}
= \frac{\theta(\sigma)-\alpha^{L}}{\theta(\sigma)\,\alpha^{H}-\alpha^{L}\,\sigma\,\theta(\sigma)}.
\end{equation}

The equilibrium as a function of morality $\kappa$ is characterized as follows:

\begin{enumerate}
    \item \textbf{Pooling at rents $(g=0)$.}  
    When $\kappa < \kappa_{\min}^{H}$, even high–$\alpha$ elites find provision unattractive 
($\Delta(\alpha^{H}\mid p=\alpha^{H}) < 0$). Both types therefore choose $g=0$, and the equilibrium outcome is pooling at rents,
with tax base $T_{0}$ given by equation~\eqref{eq:T0}.

    \item \textbf{Separation (credible reform).}
    If $\kappa_{\min}^{H} \le \kappa < \kappa_{\max}^{L}$, there exists a separating equilibrium in which
    $g^{*}(\alpha^{H})=1$ and $g^{*}(\alpha^{L})=0$, 
    and beliefs satisfy $p(g=1)=\alpha^{H}$ and $p(g=0)=\alpha^{L}$ by Bayes' rule.
    Upon observing $g=1$, citizens infer that $\alpha$ is high and update compliance accordingly, 
    leading to an immediate expansion of the tax base from $T_{0}$ (equation~\eqref{eq:T0}) 
    to $T_{1}(\alpha^{H})$ (equation~\eqref{eq:T-belief}).

    \item \textbf{Pooling at provision $(g=1)$.}  
    When $\kappa \ge \kappa_{\max}^{L}$, the low-value elite prefers to mimic and separation collapses. 
    A pooling equilibrium with both types providing exists if the high--type incentive constraint 
    under pooling is satisfied.  
    Let the on--path belief be $\bar{\alpha} = \rho\,\alpha^{H} + (1-\rho)\,\alpha^{L}$. 
    If $\alpha^{H}\sigma \le \bar{\alpha}$, pooling at provision obtains whenever 
    $\kappa \ge \max\{\kappa^{\text{pool}},\,\kappa^{H,\min}\}$, where
    \begin{equation}\label{eq:kappa-pool-summary}
    \kappa^{\text{pool}}
    = \frac{\theta(\sigma)-\alpha^{L}}{\theta(\sigma)\,[\bar{\alpha}-\alpha^{L}\sigma]},
    \qquad
    \kappa^{H,\min}
    = \frac{\theta(\sigma)-\alpha^{H}}{\theta(\sigma)\,[\bar{\alpha}-\alpha^{H}\sigma]}.
    \end{equation}
    This threshold is tight when $q^{H}=q^{L}$. When 
    $q^{H}>q^{L}$, it is an upper bound on the morality required 
    for pool-at-provision, because continuation-value costs from 
    the shifting belief path deter deviation. 
    No pooling equilibrium exists when $\alpha^{H}\sigma > \bar{\alpha}$.
\end{enumerate}
\end{proposition}

The separating equilibrium has an immediate fiscal consequence: 
credible provision triggers a within-period expansion of the 
tax base.

\begin{corollary}[Same-period fiscal expansion]\label{cor:fiscal-jump}
Within the separating region $\kappa_{\min}^{H}\le \kappa < \kappa_{\max}^{L}$, in the D1-refined separating equilibrium, observing provision $(g=1)$ leads citizens to infer $\alpha = \alpha^{H}$ and immediately increase compliance.
The tax base rises within the same period from $T_0$ in equation~\eqref{eq:T0} 
to $T_1(\alpha^{H})$ in equation~\eqref{eq:T-belief}. 
Equivalently, the same-period fiscal multiplier is
\begin{equation}\label{eq:jump-factor}
J(\kappa;\sigma)
= \frac{T_1(\alpha^{H})}{T_0}
= \frac{1-\kappa\,\sigma\,\theta(\sigma)}{1-\kappa\,\alpha^{H}} > 1.
\end{equation}
\end{corollary}

\begin{corollary}[Uniqueness in the activation region]\label{cor:uniqueness}
Suppose $\alpha^{L} < \sigma\theta(\sigma) < \underline{\alpha} < \alpha^{H} \le \theta(\sigma)$
and $\kappa \in (\kappa^{H}_{\min},\,\kappa^{L}_{\max})$. The separating
equilibrium of Proposition~\ref{prop:weak-high} is the unique D1-refined
Markov Perfect Equilibrium of the dynamic fiscal signaling game.
\end{corollary}

Morality plays a dual role. 
First, it activates reform. 
When $\kappa < \kappa_{\min}^{H}$, even high-value elites find provision too costly, so both types choose rents and fiscal capacity remains low. 
Once $\kappa$ exceeds $\kappa_{\min}^{H}$ but remains below $\kappa_{\max}^{L}$, the high-value elite's decision to provide can be credible. 
In the D1-refined separating equilibrium, the low-value elite refrains from mimicking and the improvement in fundamentals is revealed through behavior. 
Economically, the key mechanism is that moral preferences can make elites’ actions informative about fundamentals, so that citizens update beliefs and compliance immediately when reforms are genuinely warranted. This separating logic is particularly relevant for episodes in which credible public-good provision is interpreted as evidence of improved fiscal fundamentals and triggers a rapid compliance response.

Second, morality amplifies the fiscal response. 
Within the separating region $[\kappa_{\min}^{H},\,\kappa_{\max}^{L})$, the same-period increase in the tax base, captured by $J(\kappa;\sigma)$, rises monotonically with $\kappa$ and exceeds one for any $\kappa>0$. 
This reflects the endogenous compliance channel unique to {Homo Moralis} preferences: citizens internalize the universalized return to provision, which rises with $\alpha$, so credible reforms trigger an immediate and stronger fiscal expansion than under reciprocity-based mechanisms.

From equations~\eqref{eq:kappa-min-max}, stronger institutions reduce the level of morality required for credible reform. 
For parameter values such that $\alpha^{H}>\tfrac{1}{2}$ and $\alpha^{H}+\alpha^{L}>1$, 
both morality thresholds $\kappa_{\min}^{H}$ and $\kappa_{\max}^{L}$ decrease as institutional quality improves. 
This implies that when institutions are more cohesive, elites can sustain credible provision with weaker moral support, and citizens can interpret policy actions as more reliable signals of high public-good value.

\subsection{Strong high state: fundamentals-driven reform}

We now turn to the case in which the value of public goods in the high state is sufficiently large for provision to be privately optimal even in the absence of morality. 
Formally, assume
\begin{equation}
\alpha^L < \sigma\,\theta(\sigma) < \theta(\sigma) < \alpha^H,
\end{equation}
so that high-value elites find provision attractive even when $\kappa = 0$. 
Since $\alpha^{H}>\theta(\sigma)>\underline{\alpha}$, the admissibility condition $\alpha^{H}>\underline{\alpha}$ is automatically satisfied. 
In this region, morality is no longer a precondition for credible provision but continues to amplify its fiscal effects by increasing compliance and expanding the tax base.

\begin{proposition}[Strong provision incentives and equilibrium characterization]\label{prop:strong-high}
Suppose $\alpha^{L} < \sigma\,\theta(\sigma) < \theta(\sigma) < \alpha^{H}$. 
Then the equilibrium provision pattern as a function of $\kappa$ is as follows:

\begin{enumerate}
    \item \textbf{Separation (fundamentals-driven reform).}  
    For all $\kappa$ such that
\begin{equation}
0 \le \kappa < \kappa_{\max}^{L}
= \frac{\theta(\sigma)-\alpha^{L}}{\theta(\sigma)\,\alpha^{H}
-\alpha^{L}\,\sigma\,\theta(\sigma)}.
\end{equation}
    the high-value elite strictly prefers provision while the low-value elite does not.  
    Citizens interpret $g=1$ as evidence of high value, forming beliefs $p(g=1)=\alpha^{H}$ and $p(g=0)=\alpha^{L}$.  
    Morality reinforces the credibility of this signal but is not required for reform to occur.

    \item \textbf{Pooling at provision $(g=1)$.}  
    When $\kappa \ge \kappa_{\max}^{L}$, the low-value elite also finds provision profitable and separation collapses.  
    Both types provide, and $g=1$ becomes uninformative.  Let $\bar{\alpha} = \rho\,\alpha^{H} + (1-\rho)\,\alpha^{L}$ denote citizens’ on-path belief.
    The equilibrium is pooling at provision with tax base
    \begin{equation}\label{eq:Tbar}
    % R2 (w=1). Original: \frac{w c}{4\,[1-\kappa\,\bar{\alpha}]}
    T_{1}(\bar{\alpha}) = \rev{\frac{c}{4\,[1-\kappa\,\bar{\alpha}]}}.
    \end{equation}
\end{enumerate}
\end{proposition}

%\begin{proof}
%See \ref{app:proof-strong-high}.
%\end{proof}

In the strong-high state, strong provision incentives alone ensure that reform takes place even when $\kappa = 0$, so morality is no longer a precondition for credible provision. 
Higher levels of morality, however, continue to amplify the fiscal impact of reform by raising compliance and expanding the tax base, consistent with the positive slope of the fiscal multiplier $J(\kappa;\sigma)$ in equation~\eqref{eq:jump-factor}. 

The separating region is narrower than in the weak-high state 
because the high-value elite's incentive constraint is 
automatically satisfied, leaving only the low-value elite's 
constraint binding. As a result, moral preferences play a 
reinforcing rather than activating role: they magnify the fiscal 
gains of reform and enhance compliance but are not required for 
credibility. Morality and strong provision incentives are 
therefore complementary sources of fiscal capacity, with 
morality strengthening the revenue response that strong 
incentives make possible, and pooling at provision obtains for 
a wider range of morality levels.

\subsection{Robustness and extensions}\label{sec:robustness}

We address the role of equilibrium selection in sustaining 
separation and sketch an extension in which the average 
degree of morality in the population is itself an 
evolutionary outcome.

\paragraph*{Equilibrium selection and off-path beliefs}
The uniqueness of the separating outcome 
(Corollary~\ref{cor:uniqueness}) depends on citizens' off-path 
beliefs in alternative equilibrium candidates. The binding case 
is the pool-at-rents configuration, in which a deviation to 
$g=1$ is genuinely off path: without a refinement, pessimistic 
off-path beliefs $p(g=1)=\alpha^{L}$ would sustain pool-at-rents 
even inside the activation region 
$[\kappa_{\min}^{H},\kappa_{\max}^{L})$.

Throughout we report D1-refined MPEs to discipline these 
off-path beliefs. By the single-crossing property of 
$\Delta(\alpha\mid p)$ in $\alpha$, the high-value elite 
benefits strictly more from any off-path deviation to $g=1$ 
than the low-value elite does, so D1 uniquely assigns 
$p(g=1)=\alpha^{H}$ 
\citep{banks1987equilibrium,cho1990strategic}. In the 
pool-at-rents candidate, this belief makes deviation profitable 
for the high type whenever $\kappa>\kappa_{\min}^{H}$, ruling 
out pool-at-rents within the activation region.

The refinement is applied stage by stage: D1 disciplines beliefs 
at each period's signaling subgame. In the pool-at-provision 
configuration, D1 is silent on the off-path action $g=0$ because 
the elite's deviation payoff is belief-independent. We adopt the
pessimistic convention $p(g{=}0)=\alpha^{L}$, which yields the 
most demanding sufficient condition for pool-at-provision and 
coincides with the selection a dynamic-D1 argument based on 
continuation-value differences would produce 
(\ref{app:stage-reduction}, Case~3).
Without imposing a refinement, additional Markov equilibria may be supported 
by alternative off-path beliefs. The same-period fiscal jump
applies in any equilibrium in which provision is interpreted as 
evidence of the high-value state.\footnote{The weaker Intuitive 
Criterion of \citet{cho1987signaling} yields the same selection 
here. D1 is therefore a conservative choice of refinement.}

\paragraph*{Endogenous morality}
The characterization above takes $\kappa$ as a population 
parameter. \ref{app:cultural} shows that a positive 
level of morality can be sustained endogenously once the 
within-period signaling game is embedded in a long-run cultural 
process. Suppose the population consists of 
\emph{materialists} ($\kappa = 0$) and \emph{universalizers} 
($\kappa = \kappa_H > 0$), with share $\mu_n$ of 
universalizers at the start of period $n$ and aggregate 
morality $\kappa_n \equiv \mu_n\kappa_H$. Because the 
optimal report in Lemma~\ref{lem:optimal_report} is linear 
in $\kappa$, the within-period characterization of this 
section carries over with $\kappa_n$ in place of $\kappa$. 
Each period now ends with a cultural-transmission stage in 
which agents revise their type sporadically, switching 
toward whichever type has higher realized per-period utility 
at the current $\mu_n$ \citep{sethi2001}. The resulting dynamic is a standard
discrete-time replicator \citep{sandholm2010}.

\ref{app:cultural} establishes two results. First, 
the population converges from any interior initial condition 
to a stable steady state with $\mu^* \in (0,1)$ and 
$\kappa^* = \mu^*\kappa_H > 0$, providing a 
cultural-evolutionary rationale for the assumption 
$\kappa > 0$ maintained throughout.\footnote{In the 
quadratic-cost benchmark, $\mu^* = 1/2$. The precise value 
reflects the curvature of the misreporting-cost function.
The interior, globally stable character of the steady state 
is robust across specifications.} Global convergence to an interior $\mu^*$ is a distinctive feature of universalization: in the reciprocity dynamics of \citet{besley2020}, by contrast, the cultural state tips to a boundary depending on initial conditions. Second, when the elite's allocation is endogenous, the long-run fiscal regime is determined by where the steady-state morality $\kappa^*$ falls relative to the separating region of Proposition~\ref{prop:weak-high}: pooling at rents, D1-refined separation, or pooling at provision (Corollary~\ref{cor:endogenous-credibility}). Cultural evolution pins down the moral share, and the level of $\kappa_H$ selects the fiscal regime.

\section{Discussion}\label{sec:discussion}

This section situates the model relative to reciprocity-based 
alternatives, discusses the behavioral and empirical foundations 
for moral universalization, and examines the theory's scope 
through historical illustrations and identification strategies.

\paragraph*{Relation to reciprocity-based models}
The most closely related framework is \citet{besley2020}, in 
which civic-minded citizens condition compliance on the observed 
gap between public spending and elite rents. The key distinction 
is informational. In a signaling environment where fundamentals 
are private information, reciprocity rewards the observed 
allocation $g$ while universalization rewards the state $\alpha$ 
that the allocation reveals. The within-period fiscal multiplier 
under reciprocity depends on preference parameters but is 
invariant to inferred fundamentals. Under universalization it
inherits them, so credible provision triggers an additional 
compliance response through belief updating 
(\ref{app:reciprocity}, Remark~\ref{rmk:recip-structure}). When 
the two frameworks are matched on a common environment, 
universalization generates a strictly larger multiplier for any 
reciprocity preference that does not implicitly perform the 
inference universalization performs 
(Lemma~\ref{lem:recip-dominance}). The two frameworks also 
differ in their cultural dynamics: universalization converges 
globally to an interior steady state 
(Section~\ref{sec:robustness}), whereas reciprocity exhibits 
tipping and path dependence \citep[Proposition~2]{besley2020}. 
Table~\ref{tab:comparison} summarizes.

\begin{table}[!htb]\centering
\footnotesize
\caption{Summary of predictions across behavioral foundations.}
\label{tab:comparison}
\begin{tabular}{@{} p{2.8cm} p{5.5cm} p{5.5cm} @{}}
\toprule
& \textbf{Static (complete information)} 
& \textbf{Dynamic (asymmetric information)} \\
\midrule
\textbf{Selfish~agents}\newline ($\kappa = 0$)
& Standard Laffer curve. Elite provides only if 
$\alpha > \theta(\sigma)$ (strong provision state). 
No morality.
& No signaling incentive. Pooling at rents unless 
$\alpha^H > \theta(\sigma)$. Reform occurs only when 
materially self-enforcing. Compliance is independent 
of beliefs about $\alpha$.
\\[0.6em]
\textbf{Reciprocity}\newline \citep{besley2020}
& Compliance responds to the spending mix $g$ rather than 
to the value of public goods $\alpha$. Civic culture 
disciplines the elite when the civic-minded share is 
high.
& Compliance responds to observed $g$ rather than to the 
inferred $\alpha$, so the within-period fiscal multiplier 
is invariant to fundamentals 
(Remark~\ref{rmk:recip-structure}). Cultural dynamics tip 
between high- and low-capacity steady states 
\citep[Proposition~2]{besley2020}.
\\[0.6em]
\textbf{Universalization}\newline (this paper)
& Moral return $\varphi(g,\alpha,\sigma)$ expands the 
Laffer curve. Elite provides even when 
$\alpha \le \theta(\sigma)$ if 
$\kappa > \bar\kappa(\alpha,\sigma)$ (weak provision 
state). Morality disciplines the elite.
& Morality \emph{activates} credible reform in the weak 
high state, where provision is not self-enforcing: 
separation requires $\kappa \ge \kappa^H_{\min}$. The 
within-period fiscal multiplier is strictly increasing in 
$\alpha^H$ (Remark~\ref{rmk:recip-structure}). Cultural 
dynamics converge globally to an interior steady state 
(Section~\ref{sec:robustness}).
\\
\bottomrule
\end{tabular}
\end{table}
\paragraph*{Foundations for moral universalization}
A central premise of the analysis is that citizens exhibit a 
positive degree of moral universalization ($\kappa > 0$).
The evolutionary foundation of 
\citet{alger2013homo,alger2016evolution} establishes that the 
evolutionarily stable $\kappa$ equals the index of assortativity 
in the matching process \citep{bergstrom2003}, which may be 
close to zero in large anonymous populations. These preferences 
admit an evolutionary microfoundation but can equally be studied 
as a preference class in their own right. The mechanism proposed
here takes $\kappa$ as a behavioral primitive, and the foundations 
developed below do not rely on assortativity.

Recent evidence from moral psychology provides direct support. 
In a series of experiments, \citet{levine2020logic} show that 
universalization reasoning, asking ``what if everyone did 
this?'', is a robust feature of human moral judgment, and 
\citet{kwon2023not} find that subjects spontaneously apply 
universalization when evaluating rule violations, judging 
individually beneficial actions as impermissible when 
universalizing them would produce bad outcomes. These findings 
connect to two well-established cognitive biases in psychology: 
the false consensus effect \citep{ross1977false}, whereby 
individuals project their own behavior onto others, and magical 
thinking \citep{shafir1992thinking}, whereby agents act as if 
their choices causally influence aggregate outcomes. Both 
generate behavior consistent with $\kappa > 0$ in large 
anonymous populations without requiring assortative matching, 
as does the diagnostic reasoning studied by 
\citet{quattrone1984causal} in the context of 
voting.\footnote{\citet{salonia2025foundation} provides an 
axiomatic foundation for universalization preferences grounded 
in choice theory, showing that Homo Moralis and Roemer's 
Kantian equilibrium \citep{roemer2010kantian,roemer2019cooperate} 
are both special cases of a general class of universalization 
preferences. Interior $\kappa$ interpolates between selfish
(Nash) and fully universalizing (Kantian) behavior.}

In experimental economics, the behavioral patterns are 
consistent with universalization, though not exclusively so. 
Anonymous public goods games document conditional cooperation 
\citep{fischbacher2001people,fischbacher2010social,herrmann2009measuring,chaudhuri2011sustaining}, 
a pattern compatible with several behavioral foundations, 
including selfish preferences under beliefs about others' 
behavior. Universalization provides a preference-based 
microfoundation for the same behavioral pattern that does not 
require beliefs about peer behavior: the citizen conditions on a 
counterfactual in which others act as she does, rather than on 
observed peer actions. This distinction matters in atomistic 
settings such as anonymous tax compliance, where peer reports 
are unobservable. Such beliefs-based rationalizations lose force in 
atomistic settings with private actions: 
\citet{dwenger2016extrinsic} document substantial compliance 
in the tax domain under zero enforcement and without feedback 
on others' compliance, isolating a pro-social motive. 
\citet{alger2025norms} develops a model of norms and norm 
change driven by social-Kantian preferences, in which the 
Kantian component generates unconditional contribution, 
directly aligned with the universalization reasoning studied 
here. Most directly, 
\citet{vanleeuwen2024estimating} structurally estimate 
individual-level Kantian morality and report, across 112 
subjects, a mean of $0.13$ and a median of $0.10$, with 
estimates lying mostly in the interior of $[0,1]$ 
(their Table~3). Our main-text specification treats $\kappa$ 
as homogeneous across citizens. Since the compliance response
is linear in $\kappa$, aggregate fiscal capacity in any 
heterogeneous extension depends only on the population mean 
(footnote following Lemma~\ref{lem:optimal_report}), which is 
the natural object of comparison.

Finally, the cultural-dynamics extension in 
Section~\ref{sec:robustness} shows that once a positive share of 
universalizers exists, the trait is sustained endogenously 
through socialization, with the population converging to a 
stable interior steady state.

\paragraph*{Historical illustrations}
The model's most distinctive prediction, the jump-start 
dynamic in which moral norms enable credible fiscal reform 
under weak institutions, maps onto specific historical 
configurations. Two cases illustrate the relevant mechanism.

The Nordic countries are often cited as exemplars of high 
compliance under moderate penalties. Scandinavian societies 
sustain unusually high levels of generalized trust and civic 
participation, rooted in pre-existing equality and impartial 
government that preceded the construction of modern welfare 
states \citep{rothstein2005all,uslaner2008corruption}.\footnote{A 
representative cultural marker: in a 1994 SVT interview, Swedish 
minister Mona Sahlin stated that for a Social Democrat it is 
\textit{h\"aftigt} (``cool'') to pay taxes, adding that ``tax is 
the finest expression of what politics is.'' The remark has since 
been widely cited as emblematic of Nordic tax attitudes.} The 
model rationalizes this pattern: when $\kappa$ is high and 
public spending is perceived as productive (high $\alpha$), 
compliance becomes largely self-enforcing, reducing the 
marginal role of formal enforcement. This corresponds to the 
strong provision state in Proposition~\ref{prop:Elite}, where 
morality and institutional quality are complements.

Botswana provides a more direct illustration of the 
jump-start dynamic. At independence in 1966, formal state 
institutions were minimal, as the British protectorate had
invested almost nothing in administrative capacity 
\citep{acemoglu2003african}. Yet the Tswana had a 
well-developed tradition of participatory governance through 
the \textit{kgotla}, a village assembly in which the chief 
was required to seek consensus before issuing binding decisions 
\citep{acemoglu2003african,robinson2006political}. When diamond 
revenues created a sudden and large increase in the potential 
value of public spending (a positive shock to $\alpha$), the 
post-independence elite channeled these revenues into public 
goods and infrastructure rather than private extraction, a 
pattern that persisted for decades and stands in sharp contrast 
to other resource-rich African states. In the language of the 
model, Botswana at independence approximated the weak-high
state: $\alpha$ had risen sharply, $\sigma$ remained low, but 
pre-existing civic norms ($\kappa > 0$) enabled the elite to 
credibly commit to provision, triggering sustained fiscal 
expansion. This interpretation is consistent with the emphasis 
on pre-colonial institutional persistence in 
\citet{acemoglu2003african} and 
\citet{robinson2006political}\footnote{\citet{michalopoulos2013pre,michalopoulos2015ethnic} 
document the persistence of pre-colonial institutional traits 
across African ethnicities without separating formal 
institutional channels from cultural-norm channels. Because 
formal state capacity was minimal at Botswana's independence, 
the persistent trait operative in 1966 maps to $\kappa$ 
rather than $\sigma$ in the model.}, but adds a behavioral 
channel through which civic norms translate into fiscal 
outcomes via the compliance response.\footnote{\citet{papaioannou2020comment} plots state 
capacity against national versus ethnic identification across 
African countries (Figure~2 therein). Botswana appears as a 
positive outlier, with high state capacity and high national 
identification, consistent with the civic norm channel 
emphasized here. Papaioannou calls for extending Besley's 
framework to model fragility traps in low-income countries.
Our model provides one such extension, and 
Section~\ref{sec:robustness} (Endogenous morality) sketches how 
pre-existing civic norms persist endogenously through cultural 
evolution.}
\paragraph*{Identification and the correlation between morality and institutions}
A natural concern is that moral universalization and 
institutional quality may be positively correlated across 
countries: wherever $\kappa$ is high, $\sigma$ may also be 
high, making it difficult to disentangle the moral channel 
from institutional enforcement.
Three features of the model help address this concern.
First, the theory's most distinctive prediction is not about the \emph{level} of compliance but about the \emph{dynamic response} to a shift in fundamentals.
The jump-start mechanism predicts that, conditional on weak institutions, societies with stronger moral norms should exhibit larger fiscal responses to improvements in the perceived value of public goods.
This conditional prediction generates testable variation even when $\kappa$ and $\sigma$ are correlated in the cross section: identification comes from the interaction between moral norms and shocks to $\alpha$, not from either margin alone.
Equivalently, the moral channel is identifiable in the intermediate fundamentals regime: in the strong provision state ($\alpha > \theta(\sigma)$) the elite provides regardless of $\kappa$, in the transfer state ($\alpha \le \sigma\theta(\sigma)$) no admissible morality can rescue provision, and only in the weak provision state ($\sigma\theta(\sigma) < \alpha \le \theta(\sigma)$) does $\kappa$ flip the allocation. Empirical tests of the jump-start mechanism should target this intermediate regime, where reform is neither materially self-enforcing nor materially infeasible.
Second, the comparative statics in Proposition~\ref{prop:Elite} show that when $\alpha > 1/2$, morality and institutions are complements ($\partial\bar\kappa/\partial\sigma < 0$): stronger institutions \emph{reduce} the morality threshold required for provision.
A positive $\kappa$--$\sigma$ correlation is therefore consistent with the model and does not undermine the moral channel. Rather, it reflects the complementarity the model predicts.
Third, the model implies a sharp exclusion restriction at the compliance margin. In the selfish limit $\kappa = 0$, the moral return $\kappa\,\varphi(g,\alpha,\sigma)$ vanishes. With atomistic citizens in a continuum population, the material transfer $\sigma\,\theta(\sigma)(1-g)T$ is taken as given and drops out of the FOC, so conditional on the observed policy $(t,g)$, the citizen's optimal report in Lemma~\ref{lem:optimal_report} is independent of $\sigma$. Any cross-country compliance response to $\sigma$, holding $t$, $g$, and $\alpha$ fixed, therefore operates through moral internalization, providing an empirical handle that separates the two channels even when their levels are correlated.
What the model adds is that $\kappa$ has independent explanatory power \emph{conditional} on $\sigma$, a prediction amenable to empirical test using cross-country data on tax morale, enforcement capacity, and fiscal responses to reform episodes.

\section{Conclusion}\label{sec:conclusion}

Why do some states sustain high fiscal capacity under weak enforcement while others remain trapped in low-compliance equilibria?
This paper offers an answer grounded in moral internalization: when citizens evaluate their compliance decisions under universalized payoffs, the resulting fiscal capacity depends not only on enforcement but also on the degree of moral universalization $\kappa$ and the effective moral return $\varphi(g,\alpha,\sigma)$, which responds to both the value of public goods $\alpha$ and institutional quality $\sigma$.
The static model shows that morality expands the feasible tax base and can discipline elites into provision even when material incentives favor rent extraction.
The dynamic model identifies a jump-start mechanism through which morality enables credible reform: when the value of public goods increases, sufficiently moral citizens generate a compliance response that makes provision a credible signal, triggering an immediate fiscal expansion.
The framework is deliberately parsimonious: a single behavioral parameter $\kappa$ interacts with standard fiscal primitives to generate these results.
\paragraph*{Empirical and policy implications}
The theory predicts that tax morale and compliance should be 
more responsive to improvements in the perceived value of public 
spending than to changes in deterrence alone, consistent with 
experimental evidence that compliance responds to the perceived 
allocation of tax revenues 
\citep{alm1993fiscal,montenbruck2023fiscal}. The jump-start 
dynamic further predicts that fiscal capacity should respond 
discontinuously to improvements in the value of public spending 
in societies with sufficiently strong moral internalization, a 
pattern testable using reform episodes in countries with varying 
levels of tax morale. For policymakers, the complementarity between moral 
internalization and institutional quality suggests that 
investments in fiscal legitimacy, such as transparent 
reporting of public spending outcomes, can raise compliance 
more cost-effectively than enforcement alone. These insights open the door for empirical 
tests linking survey-based measures of moral universalism, 
perceptions of government effectiveness, and observed tax 
capacity.

\paragraph*{Directions for future research} \rev{We now consider how income heterogeneity might be incorporated into the model. For general income $w$, the misreporting cost is $C(d)=\tfrac{1}{2}d^2$ in the deviation $d=\tilde{w}-w$, so it penalizes the amount concealed rather than income itself (see footnote~\ref{fn:wnorm}).}
As a result, the optimal deviation from truthful reporting in Lemma~\ref{lem:optimal_report} is independent of $w$, and wealth heterogeneity does not generate cross-sectional variation in compliance behavior.
This is a tractable benchmark but not the only natural specification.
If the misreporting cost were instead proportional to income, for example, $C(d,w) = \tfrac{1}{2}(d/w)^2$, penalizing the concealment \emph{rate} rather than the concealment \emph{level}, then wealthier citizens would face a lower marginal cost of concealing a given amount, but would also internalize a larger fiscal externality under universalization, since their hypothetical aggregate impact on the tax base scales with income.
This tension could generate non-trivial cross-sectional predictions: under universalization, high-income citizens may comply \emph{more} than selfish models predict, precisely because their universalized counterfactual involves a larger revenue loss.
Whether this ``pivotality'' channel reinforces or attenuates the aggregate capacity effects identified here is an open question that merits further analysis, particularly in light of empirical evidence that compliance behavior and the forms of noncompliance vary systematically across income groups \citep[see, e.g.,][]{slemrod2007cheating}.

Taken together, the analysis suggests that moral 
universalization deserves consideration alongside enforcement 
and reciprocity as a foundation for fiscal capacity, and that 
its distinctive predictions are amenable to empirical test.

%TC:ignore
\appendix
\renewcommand{\theproposition}{\Alph{section}.\arabic{proposition}}
\renewcommand{\thetheorem}{\Alph{section}.\arabic{theorem}}
\renewcommand{\thelemma}{\Alph{section}.\arabic{lemma}}
\renewcommand{\thedefinition}{\Alph{section}.\arabic{definition}}
\section{Proofs}
\footnotesize
\subsection{Proof of Proposition~\ref{prop:Elite}}\label{app:proof-elite-alignment}
\begin{proof}
Let $\theta \equiv \theta(\sigma)$ and $s \equiv \sigma\,\theta(\sigma)$. 
Under Laffer-maximizing taxation, per-capita revenue is
\begin{equation}
T(\hat{t}, g, \kappa, c, \sigma)
% R2 (w=1). Original: \frac{w c}{4(\cdots)}
= \rev{\frac{c}{4\left(1 - \kappa\,\varphi(g,\alpha,\sigma)\right)}},
\qquad 
\varphi(g,\alpha,\sigma) = g\,\alpha + (1-g)\,s.
\end{equation}
The elite’s objective can then be written as
\begin{equation}
V(g)
= T(\hat{t}, g, \kappa, c, \sigma)\,[\alpha\, g + \theta(1-g)].
\end{equation}

Since $V(g)$ is a ratio of affine functions in $g$, it is single-peaked, so it suffices to compare the two corners, $g \in \{0,1\}$:
\begin{align}
% R2 (w=1). Original numerators: \frac{w c}{4\,(\cdots)}
V(1) &= \rev{\frac{c}{4\,(1-\kappa\,\alpha)}}\,\alpha,\\
V(0) &= \rev{\frac{c}{4\,(1-\kappa\,s)}}\,\theta.
\end{align}
Because $\kappa < 1/\alpha$ and $s < 1$, both denominators are positive. 
Hence
\begin{equation}\label{eq:key_ineq-aligned}
V(1) \ge V(0)
\;\Longleftrightarrow\;
\frac{\alpha}{1-\kappa\,\alpha}
\;\ge\;
\frac{\theta}{1-\kappa\,s}
\;\Longleftrightarrow\;
\alpha - \theta 
\;\ge\;
\kappa\,(\alpha s - \theta\,\alpha)
= -\,\kappa\,\alpha\,\theta\,(1-\sigma).
\end{equation}

\begin{enumerate}
    \item If $\alpha > \theta$, the left-hand side of~\eqref{eq:key_ineq-aligned} is positive and the right-hand side is nonpositive for all $\kappa$, 
    so $V(1) > V(0)$ and $g^* = 1$. 
    This corresponds to the \emph{strong provision state}.

    \item If $\alpha \le s = \sigma\,\theta(\sigma)$, then $\alpha - \theta < 0$ while the right-hand side is nonnegative, 
    so $V(1) < V(0)$ and $g^* = 0$. 
    This corresponds to the \emph{transfer state}.

    \item If $s < \alpha \le \theta$, both sides of~\eqref{eq:key_ineq-aligned} depend on $\kappa$. 
    Solving for indifference yields
    \begin{equation}
    \bar\kappa(\alpha,\sigma)
    = \frac{\theta - \alpha}{\alpha\,\theta\,(1-\sigma)}.
    \end{equation}
    For $\kappa < \bar\kappa(\alpha,\sigma)$, the inequality is reversed and $g^* = 0$.
    For $\kappa > \bar\kappa(\alpha,\sigma)$, we have $g^* = 1$. 
    This defines the \emph{weak provision state}.
\end{enumerate}

Finally, differentiating $\bar\kappa(\alpha,\sigma)$ gives
\begin{align}
\frac{\partial \bar\kappa}{\partial \alpha}
&= -\frac{\theta}{\alpha^2\,\theta\,(1-\sigma)} = -\frac{1}{\alpha^2\,(1-\sigma)} < 0,\\[0.5em]
\frac{\partial \bar\kappa}{\partial \sigma}
&= \frac{1/\alpha - 2}{(1-\sigma)^2}.
\end{align}
Thus $\partial \bar\kappa/\partial \sigma < 0$ if $\alpha > 1/2$, 
$\partial \bar\kappa/\partial \sigma = 0$ if $\alpha = 1/2$, 
and $\partial \bar\kappa/\partial \sigma > 0$ if $\alpha < 1/2$. 
Higher $\alpha$ always reduces the morality required for provision, 
and stronger institutions do so when the value of public spending is sufficiently high ($\alpha > 1/2$).

\end{proof}

\subsection{Off-path beliefs and the D1 refinement}\label{app:D1}

In the within-period signaling game, citizens form a posterior 
$p\in[\alpha^{L},\alpha^{H}]$ about the elite's type after observing 
$g_n\in\{0,1\}$. On path, $p$ is pinned down by Bayes' rule. For 
actions played with zero probability in equilibrium, we restrict 
beliefs using the D1 criterion \citep{banks1987equilibrium}.

\paragraph*{D1 criterion} For an off-path signal $g$, define the 
deviation-profitability sets
\begin{align*}
D^{0}(\alpha;g) &\equiv \{p\in[\alpha^{L},\alpha^{H}]:\ 
U_E(\alpha;g\mid p)\ \ge\ U_E^{*}(\alpha)\},\\
D(\alpha;g) &\equiv \{p\in[\alpha^{L},\alpha^{H}]:\ 
U_E(\alpha;g\mid p)\ >\ U_E^{*}(\alpha)\},
\end{align*}
where $U_E^{*}(\alpha)$ is the equilibrium payoff of type $\alpha$. 
D1 requires off-path beliefs to place zero weight on any type 
$\alpha$ for which there exists another type $\alpha'$ with 
$D^{0}(\alpha;g)\subsetneq D(\alpha';g)$.

\paragraph*{Verification: pool-at-rents (off-path $g=1$)}
If $\gamma_E(\alpha^{L})=\gamma_E(\alpha^{H})=0$, then 
$U_E^{*}(\alpha)=T_0\,\theta(\sigma)$, which is $\alpha$-independent. 
The deviation payoff at belief $p$ is 
$U_E(\alpha;1\mid p)=\alpha\,T_1(p)$ with 
$T_1(p)=\rev{c/[4(1-\kappa p)]}$, strictly increasing in both $\alpha$
and $p$. The threshold
\[
p^{*}(\alpha)\ \equiv\ 
\inf\{p\in[\alpha^{L},\alpha^{H}]:\ \alpha\,T_1(p)\ge T_0\,\theta(\sigma)\}
\]
is therefore strictly decreasing in $\alpha$, so 
$p^{*}(\alpha^{H})<p^{*}(\alpha^{L})$ whenever the thresholds lie 
in the interior of $[\alpha^{L},\alpha^{H}]$. Hence
\[
D^{0}(\alpha^{L};1)
\;=\;[p^{*}(\alpha^{L}),\alpha^{H}]
\;\subsetneq\;(p^{*}(\alpha^{H}),\alpha^{H}]
\;=\;D(\alpha^{H};1),
\]
with the containment trivial when $D^{0}(\alpha^{L};1)=\emptyset$. 
D1 therefore deletes $\alpha^{L}$ from the support of off-path 
beliefs at $g=1$, yielding $p(g=1)=\alpha^{H}$.

\paragraph*{Pool-at-provision (off-path $g=0$)}
If $\gamma_E\equiv 1$, the deviation payoff $T_0\,\theta(\sigma)$ 
is belief-independent. In any pool-at-provision equilibrium, 
neither type strictly prefers to deviate, so 
$D(\alpha^{L};0)=D(\alpha^{H};0)=\emptyset$ and the strict 
set-inclusion condition $D^{0}(\alpha;0)\subsetneq D(\alpha';0)$ 
cannot hold. Stage-game D1 is therefore silent. We adopt the 
convention $p(g=0)=\alpha^{L}$ throughout, which coincides with 
the dynamic-D1 selection based on continuation-value differences 
(see the proof of Lemma~\ref{lem:stage-reduction}, Case~3) and 
tightens the sufficient condition for pool-at-provision.

\paragraph*{Separating candidate}
In the separating candidate 
$(\gamma_E(\alpha^{H})=1,\ \gamma_E(\alpha^{L})=0)$, both actions 
are played with positive probability. Bayes' rule delivers 
$p(g=1)=\alpha^{H}$ and $p(g=0)=\alpha^{L}$ directly, and the D1 
refinement plays no role: the low-type IC 
\eqref{eq:lowtype-IC-D1} is evaluated at the Bayesian posterior 
$p=\alpha^{H}$.

\medskip
\noindent\emph{Remark (sensitivity of the low-type IC to $p$).} 
For completeness, we record how the low-type IC would vary with 
an arbitrary posterior $p\in[\alpha^{L},\alpha^{H}]$ assigned 
after $g=1$ (whenever $1-\kappa p>0$):
\begin{equation}\label{eq:lowtype-IC-generalp-global}
\Delta(\alpha^L\mid p)<0
\ \Longleftrightarrow\
\kappa \;<\; \kappa_{\max}^{L}(p)
\ \equiv\
\frac{\theta - \alpha^L}{\theta\,p - \alpha^L\,s}.
\end{equation}
The numerator $\theta-\alpha^{L}>0$ is independent of $p$. The
denominator $\theta p - \alpha^{L} s$ is strictly increasing in 
$p$. Hence $\kappa_{\max}^{L}(p)$ is strictly decreasing in $p$, 
and the Bayesian posterior $p=\alpha^{H}$ yields the most 
demanding (smallest) upper bound on $\kappa$. Any lower posterior 
would widen the separating interval.

\subsection{Proof of Lemma~\ref{lem:stage-reduction}}\label{app:stage-reduction}
\begin{proof}
For any fixed belief $p_n$, the per-period payoff 
$U_E(\alpha_n;g_n\mid p_n)$ is convex in $g_n$, so any 
per-period best response satisfies $g_n\in\{0,1\}$. The
optimal tax rate is the Laffer rate and is suppressed from 
the notation. Without loss of generality we therefore 
restrict attention to elite strategies 
$\gamma_E:\{\alpha^{L},\alpha^{H}\}\to\{0,1\}$. Fix such an 
MPE, with on-path decoding $p_n(g)$ and off-path beliefs 
to be specified case by case. Let $\delta\in[0,1)$ be the elite's discount factor. 
Citizens' posterior on the current state is summarized by the 
probability $\pi_n\equiv\Pr(\alpha_n=\alpha^{H}\mid\text{history})$.
The belief entering the tax base is the induced expected value 
$p_n=\pi_n\alpha^{H}+(1-\pi_n)\alpha^{L}$, so $T_1(p_n)$ is 
strictly increasing in $\pi_n$. We verify the one-shot 
deviation principle by case, showing that the elite's 
period-$n$ choice affects the continuation value in a way 
that never overturns a strict stage-game IC.

\medskip
\noindent\textbf{Case 1: Separation.} Let 
$\gamma_E(\alpha^{H})=1$ and $\gamma_E(\alpha^{L})=0$. On path, 
$g_n$ fully reveals $\alpha_n$, so $p_n=\alpha_n$ by Bayes' rule. 
Both actions are played with positive probability, so the D1 
refinement plays no role in this candidate.

In every subsequent period $n+k$, $k\ge 1$, the elite plays 
$g_{n+k}=\gamma_E(\alpha_{n+k})$, which is again fully 
revealing. Bayesian updating with a fully revealing signal is 
prior-independent: citizens' period-$(n+k)$ posterior is 
pinned at $p_{n+k}=\alpha_{n+k}$ regardless of any prior 
inherited from period $n$. Any belief component carried 
forward from the period-$n$ history, correctly inferred or 
not, is overwritten in one step. Moreover, the rents tax 
base $T_0=\rev{c/[4(1-\kappa s)]}$ is $\alpha$-independent, so the
elite's period-$(n+k)$ payoff is 
$T_1(\alpha_{n+k})\alpha_{n+k}$ if $\gamma_E(\alpha_{n+k})=1$ 
and $T_0\,\theta(\sigma)$ if $\gamma_E(\alpha_{n+k})=0$, 
depending only on the realized state.

The continuation value from $n+1$ onward therefore depends on 
the true $\alpha_n$ (through the Markov transition) and is 
independent of $g_n$. The period-$n$ IC reduces to the 
stage-game comparison $\Delta(\alpha_n\mid p_n)\gtreqless 0$, 
with $p_n=\alpha^{H}$ following $g_n=1$ and $p_n=\alpha^{L}$ 
following $g_n=0$, both by Bayes' rule. The resulting sufficient 
conditions for separation depend only on the primitives 
$(\alpha^{L},\alpha^{H},\sigma,c,w)$. Neither the transition 
probabilities $(q^{H},q^{L})$ nor any realized state enters, 
because fully revealing signals re-anchor beliefs each period.

\medskip
\noindent\textbf{Case 2: Pooling at rents.} Let 
$\gamma_E\equiv 0$. On path, $g_n$ is uninformative, so 
$\pi_n$ equals the prior at the start of period $n$.
Iterating across periods with uninformative pool-at-rents 
signals yields the stationary marginal $\pi_n=\rho$. Off 
path, $p_n(g_n=1)=\alpha^{H}$ by D1: the deviation payoff 
$T_1(p)\alpha$ is strictly increasing in both $p$ and 
$\alpha$, so the set of citizen responses for which deviation 
is profitable is strictly larger for the high type.

From period $n+1$ onward, the elite plays $g_{n+k}=0$, and 
the tax base is $T_0$ in every period. Both $T_0$ and 
$\theta(\sigma)$ are independent of $\alpha$ and of $p$, so 
the elite's per-period rents payoff $T_0\,\theta(\sigma)$ is 
invariant to citizens' belief about any previous state. The 
continuation value is independent of $g_n$, and the 
period-$n$ IC reduces to $\Delta(\alpha_n\mid p_n)\le 0$ at 
the relevant $p_n$.

\medskip
\noindent\textbf{Case 3: Pooling at provision.} Let 
$\gamma_E\equiv 1$. On path, $g_n$ is uninformative and 
$\pi_n=\rho$. For the off-path action $g_n=0$, observe that 
the citizen's universalized moral term 
$\kappa\,\varphi(0,\alpha,\sigma)=\kappa\,\sigma\,\theta(\sigma)$ 
is $\alpha$-independent, so the optimal report, and hence 
the tax base $T_0$, does not respond to citizens' belief $p$. 
The elite's deviation payoff $T_0\,\theta(\sigma)$ is 
therefore belief-independent, and stage-game D1 does not 
refine off-path beliefs on $g_n=0$.\footnote{Formally, in any 
pool-at-provision equilibrium neither type strictly prefers to 
deviate, so the strict-profitability sets $D(\alpha^{L})$ and 
$D(\alpha^{H})$ are empty and the D1 strict set-inclusion 
condition $D^{0}(\alpha)\subsetneq D(\alpha')$ cannot hold. A 
dynamic-D1 argument based on continuation-value differences 
would select $p_n(g_n=0)=\alpha^{L}$ under mild regularity 
conditions on $\delta$, but we avoid importing that 
apparatus.} We adopt the pessimistic convention 
$p_n(g_n=0)=\alpha^{L}$ (equivalently, $\pi_n=0$), which 
yields the most demanding sufficient condition for 
pool-at-provision and coincides with the D1-like intuition 
that a deviation to rents is attributed to the type with the 
greater contemporaneous gain.

Under this convention, a period-$n$ deviation fixes 
$\pi_n=0$, so citizens' prior on $\alpha_{n+1}$ is 
$\pi_{n+1}=q^{L}<\rho$. In each subsequent period, the elite 
plays $g_{n+k}=1$, which is uninformative, so
\[
\pi_{n+k+1} \;=\; \pi_{n+k}\,q^{H} + (1-\pi_{n+k})\,q^{L}, 
\qquad \pi_{n+1}=q^{L}.
\]
The affine map $\pi\mapsto q^{L}+\pi(q^{H}-q^{L})$ is a 
contraction toward the stationary fixed point 
$\rho=q^{L}/(1-q^{H}+q^{L})$. Initialized below $\rho$, the 
sequence converges monotonically from below. Hence 
$\pi_{n+k}\le\rho$ in every $k\ge 1$, with strict inequality 
when $q^{H}>q^{L}$.

Since $T_1$ is strictly increasing in $\pi$, post-deviation 
per-period compliance is weakly below its on-path level. 
Taking expectations over $\alpha_{n+k}$ conditional on the 
true current state, the per-period elite payoff after a 
deviation is weakly lower in every future period, for either 
realization of $\alpha_n$. The continuation value satisfies
\[
V_E^{\text{dev}}(\alpha) \;\le\; V_E^{\text{eq}}(\alpha),
\qquad \alpha\in\{\alpha^{L},\alpha^{H}\},
\]
with strict inequality when $q^{H}>q^{L}$ and $\delta>0$. 
Hence the stage-game IC $\Delta(\alpha\mid p)\ge 0$ 
evaluated at the on-path belief is \emph{sufficient} for 
pool-at-provision incentive compatibility. It is also 
necessary if and only if $q^{H}=q^{L}$ (the i.i.d.\ case), 
in which case belief paths are invariant to deviations and 
the continuation cost vanishes.

\medskip
\noindent Combining Cases 1--3, the sign of 
$\Delta(\alpha_n\mid p_n(g_n=1))$ determines whether provision is 
optimal for type $\alpha_n$. This stage-game characterization is 
exact in the separating and pool-at-rents configurations and 
sufficient in pool-at-provision (exact iff $q^{H}=q^{L}$). The 
separating region in 
Propositions~\ref{prop:weak-high}--\ref{prop:strong-high} is 
therefore exact. The pool-at-provision region is conservative, with 
the true region weakly larger when $q^{H}>q^{L}$.
\end{proof}

\subsection{Proof of Proposition~\ref{prop:weak-high}}\label{app:proof-weak-high}
\begin{proof}
Let $\theta \equiv \theta(\sigma)$ and $s \equiv \sigma\,\theta(\sigma)$. In the weak-high state we assume
\begin{equation}
\alpha^L < s < \underline{\alpha}(c,\sigma) < \alpha^H \le \theta,
\end{equation}
where $\underline{\alpha}(c,\sigma)=\theta[\sigma+\tfrac{c}{2}(1-\sigma)]$ is defined in Section~\ref{sec:eq-characterization}.
Under Laffer-maximizing taxation, the per-period gain from provision (relative to rents), evaluated at belief $p$, is
\begin{equation}\label{eq:Delta-app}
% R2 (w=1). Original prefactor: \frac{w c}{4}
\Delta(\alpha \mid p)
= \rev{\frac{c}{4}}\Bigg[
\frac{\alpha}{1-\kappa\,p}
-
\frac{\theta}{1-\kappa\,s}
\Bigg].
\end{equation}
Since $\rev{\tfrac{c}{4}}>0$, the sign of $\Delta(\alpha\mid p)$ is the sign of the bracket.

\paragraph{1. Separation: existence and thresholds}
Under separation, beliefs are correct on path: $p=\alpha$. Plugging $p=\alpha$ into \eqref{eq:Delta-app} and rearranging,
\begin{equation}\label{eq:sep-ineq}
\Delta(\alpha;\kappa)\ge 0
\;\Longleftrightarrow\;
\frac{\alpha}{1-\kappa\,\alpha}\;\ge\;\frac{\theta}{1-\kappa\,s}
\;\Longleftrightarrow\;
\alpha - \theta \;\ge\; \kappa\,(\alpha s - \theta \alpha)
= -\,\kappa\,\alpha\,\theta\,(1-\sigma).
\end{equation}
For the high type $\alpha=\alpha^H\le \theta$,
\begin{equation}\label{eq:kappa-minH-deriv}
\Delta(\alpha^H;\kappa)\ge 0
\ \Longleftrightarrow\
\kappa \;\ge\; \kappa_{\min}^H
\;\equiv\;
\frac{\theta - \alpha^H}{\theta\,\alpha^H\,(1-\sigma)} .
\end{equation}
For the low type under the separating candidate, $g=1$ is played on path by the high type, so citizens observing $g=1$ form the Bayesian posterior $p=\alpha^H$. The low-type IC for separation is therefore
\begin{equation}\label{eq:lowtype-IC-D1}
\Delta(\alpha^L\mid p=\alpha^H)<0
\ \Longleftrightarrow\
\frac{\alpha^L}{1-\kappa\,\alpha^H}
<
\frac{\theta}{1-\kappa\,s}
\ \Longleftrightarrow\
\kappa \;<\; \kappa_{\max}^L
\ \equiv\
\frac{\theta - \alpha^L}{\theta\,\alpha^H - \alpha^L\,s}.
\end{equation}
In the weak-high state $s<\alpha^H\le\theta$ implies $\theta\,\alpha^H - \alpha^L s>0$, so $\kappa_{\max}^L$ is well-defined and strictly positive. To verify that the separating region is nonempty, note that $\kappa_{\min}^H$ and $\kappa_{\max}^L$ can both be written as $(\theta-\alpha)/(\theta\alpha^H - \alpha s)$ evaluated at $\alpha=\alpha^H$ and $\alpha=\alpha^L$ respectively. This ratio is strictly decreasing in $\alpha$ (its derivative in $\alpha$ equals $\theta(s-\alpha^H)/(\theta\alpha^H - \alpha s)^2 < 0$ since $s<\alpha^H$), so $\kappa_{\min}^H < \kappa_{\max}^L$. Therefore separation exists for
\begin{equation}
\kappa_{\min}^H \;\le\; \kappa \;<\; \kappa_{\max}^L ,
\end{equation}
provided the interval is nonempty, which was verified above. Moreover, $\alpha^{H}>\underline{\alpha}$ ensures $\kappa_{\min}^{H}<(1-c/2)/\alpha^{H}$ (Section~\ref{sec:eq-characterization}), so the interval intersects the admissible range of Assumption~\ref{ass:kappa}. Within this interval, $\Delta(\alpha^H;\kappa)>0$ and $\Delta(\alpha^L\mid \alpha^H)<0$, so $g^*(\alpha^H)=1$ and $g^*(\alpha^L)=0$.
\paragraph{2. Pooling at rents when $\kappa < \kappa_{\min}^H$}
If $\kappa < \kappa_{\min}^H$, then \eqref{eq:kappa-minH-deriv} fails 
and $\Delta(\alpha^H;\kappa) < 0$. Hence both types strictly prefer 
$g=0$, so the unique equilibrium outcome is pooling at rents with 
tax base $T_0$ (equation~\eqref{eq:T0}). Under D1, $p(g=1)=\alpha^{H}$, 
which yields the maximum deviation payoff $\alpha\,T_1(\alpha^{H})$. 
Since $\Delta(\alpha^{H}\mid\alpha^{H})<0$ for $\kappa<\kappa_{\min}^{H}$, 
monotonicity of $\Delta(\cdot\mid\alpha^{H})$ in $\alpha$ gives 
$\Delta(\alpha^{L}\mid\alpha^{H})<0$ a fortiori, so neither type 
prefers to deviate even at the most favorable off-path belief.

\paragraph{3. Pooling at provision when $\kappa \ge \kappa_{\max}^L$: feasibility conditions}
Suppose both types choose $g=1$ and citizens’ on-path belief is
\[
\bar{\alpha} \;\equiv\; \rho\,\alpha^H + (1-\rho)\,\alpha^L .
\]
The type-$\alpha$ incentive constraint under pooling at provision is $\Delta(\alpha \mid \bar{\alpha})\ge 0$, i.e.
\begin{equation}\label{eq:IC-pool}
\frac{\alpha}{1-\kappa\,\bar{\alpha}} \;\ge\; \frac{\theta}{1-\kappa\,s}
\;\Longleftrightarrow\;
\alpha - \theta \;\ge\; \kappa\,\theta\,(\alpha \sigma - \bar{\alpha}).
\end{equation}

\emph{High type.} For $\alpha=\alpha^H \le \theta$:
\begin{equation}\label{eq:IC-H-correct}
\alpha^H - \theta \;\ge\; \kappa\,\theta\,(\alpha^H \sigma - \bar{\alpha}).
\end{equation}
There are two subcases.
\begin{itemize}
\item If $\alpha^H \sigma \le \bar{\alpha}$, then $(\alpha^H \sigma - \bar{\alpha})\le 0$ and \eqref{eq:IC-H-correct} is equivalent to the \emph{lower bound}
\begin{equation}\label{eq:kappa-Hmin}
\kappa \;\ge\; \kappa^{H,\min}
\ \equiv\
\frac{\theta - \alpha^H}{\theta\,(\bar{\alpha} - \alpha^H \sigma)} \;\ge\; 0,
\end{equation}
with equality only when $\alpha^H=\theta$. Thus the high-type IC under pooling is \emph{not automatic}: it requires sufficiently high morality.
\item If $\alpha^H \sigma > \bar{\alpha}$, then $(\alpha^H \sigma - \bar{\alpha})>0$ and \eqref{eq:IC-H-correct} would impose the \emph{upper bound} $\kappa \le \frac{\alpha^H-\theta}{\theta\,(\alpha^H \sigma - \bar{\alpha})}\le 0$, which is infeasible for $\kappa\ge 0$ unless $\alpha^H=\theta$ (knife-edge). Hence pooling at provision is infeasible in this subcase.
\end{itemize}

\emph{Low type.} For $\alpha=\alpha^L<\theta$, \eqref{eq:IC-pool} is equivalent to
\begin{equation}\label{eq:kappa-pool-deriv}
\kappa \;\ge\; \kappa^{\text{pool}}
\;\equiv\;
\frac{\theta - \alpha^L}{\theta\,(\bar{\alpha} - \alpha^L \sigma)}.
\end{equation}
Since $\bar{\alpha}>\alpha^L \sigma$ (because $\sigma<1$ and $\alpha^H>\alpha^L$), $\kappa^{\text{pool}}$ is well-defined and positive.

Therefore, if $\kappa \ge \kappa_{\max}^L$ (so the low type is willing to provide under separation), the stage-game ICs for pool-at-provision require
\begin{equation}\label{eq:pool-feasible-correct}
\alpha^H \sigma \le \bar{\alpha}
\quad\text{and}\quad
\kappa \;\ge\; \max\{\kappa^{\text{pool}},\,\kappa^{H,\min}\}.
\end{equation}
By Lemma~\ref{lem:stage-reduction}, this condition is sufficient for pool-at-provision and tight when $q^{H}=q^{L}$. When $q^{H}>q^{L}$, it is an upper bound on the required morality, because continuation-value costs from the shifting belief path deter deviation. If instead $\alpha^H \sigma > \bar{\alpha}$, the stage-game high-type IC is infeasible for any $\kappa\ge 0$ (except the knife-edge $\alpha^H=\theta$).

\paragraph{4. Summary of regions}
Combining steps 1–3:
\begin{equation}
\begin{cases}
\text{Pooling at rents $(g=0)$}, & \text{if } \kappa < \kappa_{\min}^H,\\[0.4em]
\text{Separation $(g^*(\alpha^H)=1,\ g^*(\alpha^L)=0)$}, & \text{if } \kappa_{\min}^H \le \kappa < \kappa_{\max}^L,\\[0.4em]
\text{Pooling at provision $(g=1)$}, & \text{if } \kappa \ge \kappa_{\max}^L,\ \alpha^H \sigma \le \bar{\alpha},\ \text{and }\kappa \ge \max\{\kappa^{\text{pool}},\kappa^{H,\min}\}.
\end{cases}
\end{equation}
This is the characterization in Proposition~\ref{prop:weak-high}, with the pool-at-provision conditions understood as sufficient per Lemma~\ref{lem:stage-reduction}.

Finally, collecting thresholds for reference:
\begin{equation}
\kappa_{\min}^{H}
= \frac{\theta-\alpha^{H}}{\theta\,\alpha^{H}\,(1-\sigma)},\quad
\kappa_{\max}^{L}
= \frac{\theta-\alpha^{L}}{\theta\,\alpha^{H}-\alpha^{L}\,s},\quad
\kappa^{\text{pool}}
= \frac{\theta-\alpha^{L}}{\theta\,(\bar{\alpha}-\alpha^{L}\sigma)},\quad
\kappa^{H,\min}
= \frac{\theta - \alpha^H}{\theta\,(\bar{\alpha} - \alpha^H \sigma)}.
\end{equation}
\end{proof}

\subsection{Proof of Corollary~\ref{cor:uniqueness}}\label{app:proof-uniqueness}
\begin{proof}
Fix $\kappa \in (\kappa^{H}_{\min},\,\kappa^{L}_{\max})$. We rule out
each non-separating candidate.

\paragraph{Pooling at rents} Suppose both types choose $g=0$. Under
D1, single-crossing of $\Delta(\alpha\mid p)$ in $\alpha$ implies that
any off-path $g=1$ is attributed to the high-value elite, so
$p(g=1)=\alpha^{H}$. The high-type's deviation gain is
$\Delta(\alpha^{H}\mid \alpha^{H})$, which is strictly positive for
$\kappa>\kappa^{H}_{\min}$ by~\eqref{eq:IC-H}. The high type therefore
strictly prefers to deviate, and pooling at rents is not an equilibrium.

\paragraph{Pooling at provision} Suppose both types choose $g=1$.
The on-path belief is $\bar\alpha = \rho\alpha^{H}+(1-\rho)\alpha^{L}
<\alpha^{H}$, and the low-type IC requires $\Delta(\alpha^{L}\mid
\bar\alpha)\ge 0$, equivalently $\kappa\ge\kappa^{\text{pool}}$, where
$\kappa^{\text{pool}}=(\theta-\alpha^{L})/[\theta(\bar\alpha-\alpha^{L}\sigma)]$
as defined in~\eqref{eq:kappa-pool-summary}. Since
$\kappa^{\text{pool}}$ and $\kappa^{L}_{\max}$ share the numerator
$\theta-\alpha^{L}>0$, and the denominator of $\kappa^{\text{pool}}$
equals $\theta\bar\alpha - \alpha^{L} s < \theta\alpha^{H} - \alpha^{L} s$
(the denominator of $\kappa^{L}_{\max}$), we have
$\kappa^{\text{pool}}>\kappa^{L}_{\max}$. For $\kappa<\kappa^{L}_{\max}$
the low-type IC is therefore violated, and the low type strictly prefers
to deviate to $g=0$.

\paragraph{Semi-pooling and mixed equilibria.} Consider semi-pooling in
which only the low type randomizes ($\gamma_E(\alpha^{L})\in(0,1)$,
$\gamma_E(\alpha^{H})=1$). Low-type indifference requires
$\Delta(\alpha^{L}\mid p(g=1))=0$ for some posterior $p(g=1)\in
[\alpha^{L},\alpha^{H}]$. Since $\Delta(\alpha^{L}\mid p)$ is strictly
increasing in $p$ and $\Delta(\alpha^{L}\mid\alpha^{H})<0$ for
$\kappa<\kappa^{L}_{\max}$, no such posterior exists. By the symmetric
argument, semi-pooling with only the high type randomizing requires
$\Delta(\alpha^{H}\mid\alpha^{H})=0$, which holds only at the
knife-edge $\kappa=\kappa^{H}_{\min}$. Finally, fully-mixed equilibria
would require both types' indifference conditions to pin the same
posterior $p(g=1)$, which is impossible for $\alpha^{H}\neq\alpha^{L}$.

The unique D1-refined MPE is therefore the separating equilibrium of
Proposition~\ref{prop:weak-high}.
\end{proof}

\subsection{Proof of Proposition~\ref{prop:strong-high}}\label{app:proof-strong-high}
\begin{proof}
Let $\theta \equiv \theta(\sigma)$ and $s \equiv \sigma\,\theta(\sigma)$. In the strong-high state we assume
\begin{equation}
\alpha^L < s < \theta < \alpha^H .
\end{equation}
Under Laffer-maximizing taxation, the per-period gain from provision (relative to rents), evaluated at belief $p$, is
\begin{equation}\label{eq:Delta-strong}
% R2 (w=1). Original prefactor: \frac{w c}{4}
\Delta(\alpha \mid p)
= \rev{\frac{c}{4}}\!\left[
\frac{\alpha}{1-\kappa\,p}
-
\frac{\theta}{1-\kappa\,s}
\right].
\end{equation}
Since $\rev{\tfrac{c}{4}}>0$, the sign of $\Delta(\alpha\mid p)$ coincides with the sign of the bracketed term.

\paragraph{1. Separation for small and intermediate $\kappa$}
On path under separation, beliefs are correct, so set $p=\alpha$. Rearranging,
\begin{equation}\label{eq:sep-ineq-strong}
\Delta(\alpha;\kappa)\ge 0
\;\Longleftrightarrow\;
\frac{\alpha}{1-\kappa\,\alpha}\;\ge\;\frac{\theta}{1-\kappa\,s}
\;\Longleftrightarrow\;
\alpha - \theta \;\ge\; \kappa\,(\alpha s - \theta \alpha)
= \kappa\,\alpha\,(s-\theta).
\end{equation}
For the high type, $\alpha=\alpha^H>\theta$ and $s-\theta<0$, so the right-hand side is nonpositive while the left-hand side is strictly positive. Hence
\[
\Delta(\alpha^H;\kappa)>0 \qquad \text{for all feasible } \kappa \in [0,\,1/\alpha^H).
\]
In the separating candidate, $g=1$ is played on path by the high type, so $p(g=1)=\alpha^H$ by Bayes' rule. Separation therefore requires
\[
\Delta(\alpha^L \mid p=\alpha^H;\kappa)<0
\ \Longleftrightarrow\
\frac{\alpha^L}{1-\kappa\,\alpha^H} < \frac{\theta}{1-\kappa\,s}
\ \Longleftrightarrow\
\kappa \;<\; \kappa_{\max}^{L}
\;\equiv\;
\frac{\theta-\alpha^{L}}{\theta\,\alpha^{H}-\alpha^{L}\,s},
\]
where $s<\alpha^H$ ensures the denominator is strictly positive. Therefore, for any $0 \le \kappa < \kappa_{\max}^{L}$ we have $\Delta(\alpha^H;\kappa)>0$ and $\Delta(\alpha^L\mid \alpha^H;\kappa)<0$, yielding separation with $g^*(\alpha^H)=1$ and $g^*(\alpha^L)=0$.

\paragraph{2. Pooling at provision when $\kappa \ge \kappa_{\max}^{L}$: feasibility}
Suppose $\kappa \ge \kappa_{\max}^{L}$ so the low type is willing to provide under separation. Consider pooling with $g=1$ on path and belief
\[
\bar{\alpha} \;=\; \rho\,\alpha^{H} + (1-\rho)\,\alpha^{L}.
\]
Type–$\alpha$’s IC under pooling at provision is $\Delta(\alpha \mid \bar{\alpha})\ge 0$, i.e.
\begin{equation}\label{eq:IC-pool-strong}
\frac{\alpha}{1-\kappa\,\bar{\alpha}} \;\ge\; \frac{\theta}{1-\kappa\,s}
\;\Longleftrightarrow\;
\alpha - \theta \;\ge\; \kappa\,(\alpha s - \theta\,\bar{\alpha})
= \kappa\,\theta\,(\alpha \sigma - \bar{\alpha}).
\end{equation}
\emph{High type.} For $\alpha=\alpha^H>\theta$:
\[
\alpha^H - \theta \;\ge\; \kappa\,\theta\,(\alpha^H \sigma - \bar{\alpha}).
\]
If $\alpha^H \sigma \le \bar{\alpha}$, the right-hand side is $\le 0$, so the high-type IC holds for all $\kappa \ge 0$. If $\alpha^H \sigma > \bar{\alpha}$, it imposes
\begin{equation}\label{eq:kappa-Hmax-strong}
\kappa \;\le\; \kappa^{H,\max}
\;\equiv\;
\frac{\alpha^{H}-\theta}{\theta\,(\alpha^{H}\sigma - \bar{\alpha})},
\end{equation}
which is strictly positive and finite in the strong-high state.

\emph{Low type.} For $\alpha=\alpha^L<\theta$, \eqref{eq:IC-pool-strong} is equivalent to
\begin{equation}\label{eq:kappa-pool-strong}
\kappa \;\ge\; \kappa^{\text{pool}}
\;\equiv\;
\frac{\theta-\alpha^{L}}{\theta\,(\bar{\alpha} - \alpha^{L}\sigma)},
\end{equation}
which is well-defined since $\bar{\alpha}>\alpha^L\sigma$ whenever $\rho>0$.

Therefore, a pooling-at-provision equilibrium exists iff
\begin{equation}\label{eq:pool-feasible-strong}
\kappa \ \ge\ \kappa_{\max}^L
\quad\text{and}\quad
\begin{cases}
\alpha^{H}\sigma \le \bar{\alpha} \ \ \text{and}\ \ \kappa \ \ge\ \kappa^{\text{pool}},\\[0.25em]
\alpha^{H}\sigma > \bar{\alpha} \ \ \text{and}\ \ \max\{\kappa_{\max}^L,\kappa^{\text{pool}}\} \ \le\ \kappa \ \le\ \kappa^{H,\max}.
\end{cases}
\end{equation}
(Feasibility also requires $\,\kappa<1/\alpha^H$ and $\,\kappa<1/\bar{\alpha}$.)

\paragraph{3. Conclusion}
Combining steps 1–2 yields: separation for $0\le \kappa<\kappa_{\max}^L$, and pooling at provision when $\kappa \ge \kappa_{\max}^L$ subject to \eqref{eq:pool-feasible-strong}. This completes the proof.
\end{proof}

\section{Unaligned Elites}\label{app:unaligned}

This appendix extends the static framework by allowing the elite and the citizens to value the public good differently. 
All primitives and timing remain as in Section~\ref{sec:static}, except that the $\alpha$ parameters may differ across groups. 
Let $\alpha_E$ denote the elite’s valuation of public spending and $\alpha_C$ the citizens’ valuation, with $\alpha_E \ne \alpha_C$ in general. 
This asymmetry captures situations in which elites place relatively less weight on public goods or more weight on private rents, while citizens care primarily about provision. 
Moral preferences continue to shape compliance and the effective tax base exactly as before, but now the elite’s allocation choice $g \in \{0,1\}$ trades off its own valuation $\alpha_E$ against the citizens’ response driven by $\alpha_C$.
The following proposition characterizes the elite’s optimal allocation rule.

\begin{proposition}[Elite's allocation under misalignment]
\label{prop:unaligned}\begin{enumerate}
    \item \textbf{Weak state.} If $\alpha_E \le \theta(\sigma)$ and $\alpha_C \le \sigma \theta(\sigma)$, then $g^*=0$ for all $\kappa$.

    \item \textbf{Common-interest state.} If $\alpha_E \ge \theta(\sigma)$ and $\alpha_C \ge \sigma \theta(\sigma)$, then $g^*=1$ for all $\kappa \in [0,1/\alpha_C)$.

    \item \textbf{Contested common-interest state.} If $\alpha_E > \theta(\sigma)$ and $\alpha_C < \sigma \theta(\sigma)$, the morality cutoff is
    \begin{align}
        \bar{\kappa}^{(+)}(\alpha_C,\alpha_E,\sigma)
        = \frac{\alpha_E - \theta(\sigma)}{\alpha_E \sigma \theta(\sigma) - \theta(\sigma)\alpha_C},
    \end{align}
    which satisfies $\bar{\kappa}^{(+)} \in (0,1/\alpha_C)$. Then
    \begin{align}
        g^*=\begin{cases}
            0, & \text{if } \kappa < \bar{\kappa}^{(+)},\\
            1, & \text{if } \kappa \ge \bar{\kappa}^{(+)}.
        \end{cases}
    \end{align}

    \item \textbf{Contested transfer state.} If $\alpha_E < \theta(\sigma)$ and $\alpha_C > \sigma \theta(\sigma)$, the morality cutoff is
    \begin{align}
        \bar{\kappa}^{(-)}(\alpha_C,\alpha_E,\sigma)
        = \frac{\theta(\sigma) - \alpha_E}{\theta(\sigma)\alpha_C - \alpha_E \sigma \theta(\sigma)},
    \end{align}
    which satisfies $\bar{\kappa}^{(-)} \in (0,1/\alpha_C)$. Then
    \begin{align}
        g^*=\begin{cases}
            0, & \text{if } \kappa < \bar{\kappa}^{(-)},\\
            1, & \text{if } \kappa \ge \bar{\kappa}^{(-)}.
        \end{cases}
    \end{align}
\end{enumerate}

On the knife-edge $\alpha_E \sigma \theta(\sigma) = \theta(\sigma)\alpha_C$, the elite is indifferent between $g=0$ and $g=1$.
\end{proposition}
\begin{proof}
Let $\theta\equiv \theta(\sigma)$ and $s\equiv \sigma\,\theta(\sigma)$. Under Laffer-maximizing taxation, per-capita revenue is
\begin{align}
% R2 (w=1). Original: \frac{w c}{4(\cdots)}
T(\hat{t}, g, \kappa, c, \sigma)
= \rev{\frac{c}{4\left(1-\kappa\,\varphi(g,\alpha_C,\sigma)\right)}},
\qquad
\varphi(g,\alpha_C,\sigma)= g\,\alpha_C + (1-g)\,s .
\end{align}
The elite’s objective can be written as
\begin{align}
V(g)
= T(\hat{t}, g, \kappa, c, \sigma)\,\big[\alpha_E g + \theta(1-g)\big].
\end{align}

Since $g\in\{0,1\}$ suffices (the objective is a ratio of affine functions in $g$ and thus single-peaked), compare the two corners:
\begin{align}
% R2 (w=1). Original numerators: \frac{w c}{4\,(\cdots)}
V(1) &= \rev{\frac{c}{4\left(1-\kappa\,\alpha_C\right)}}\,\alpha_E,\\
V(0) &= \rev{\frac{c}{4\left(1-\kappa\,s\right)}}\,\theta.
\end{align}
Because $\kappa<1/\alpha_C$ and $s<1$, both denominators are positive. Hence
\begin{align}
V(1)\ge V(0)
\;\Longleftrightarrow\;
\frac{\alpha_E}{1-\kappa \alpha_C}\;\ge\;\frac{\theta}{1-\kappa s}
\;\Longleftrightarrow\;
\alpha_E-\theta \;\ge\; \kappa\big(\alpha_E s - \theta \alpha_C\big).
\label{eq:key_ineq}
\end{align}

\emph{Aligned states.}
If $\alpha_E\le \theta$ and $\alpha_C\le s$, then at $\kappa=0$ the inequality~\eqref{eq:key_ineq} reduces to $\alpha_E \ge \theta$, which fails. For $\kappa>0$, if $\alpha_E s - \theta\alpha_C \ge 0$ the right-hand side is nonnegative while the left-hand side is nonpositive, so the inequality still fails. If $\alpha_E s - \theta\alpha_C < 0$, indifference would require $\kappa = (\theta-\alpha_E)/(\theta\alpha_C - \alpha_E s)$, but $\alpha_C \le s$ implies this threshold is at least $1/\alpha_C$ (since $(\theta-\alpha_E)\alpha_C \ge \theta\alpha_C - \alpha_E s \Leftrightarrow s \ge \alpha_C$). Hence $V(1)<V(0)$ for all feasible $\kappa \in [0,1/\alpha_C)$ and $g^*=0$ (weak state). 
A symmetric argument establishes the common-interest state: if $\alpha_E\ge \theta$ and $\alpha_C\ge s$, then $V(1)\ge V(0)$ at $\kappa=0$, and no feasible $\kappa$ reverses the ranking (the indifference threshold, when it exists, again exceeds $1/\alpha_C$), so $g^*=1$.

\emph{Contested states.}
When preferences are misaligned, \eqref{eq:key_ineq} delivers a morality cutoff.

\smallskip
\noindent
(i) If $\alpha_E>\theta$ and $\alpha_C<s$, then $\alpha_E s - \theta \alpha_C>0$ and the cutoff is
\begin{align}
\bar{\kappa}^{(+)}(\alpha_C,\alpha_E,\sigma)
= \frac{\alpha_E-\theta}{\alpha_E s - \theta \alpha_C}\in(0,1/\alpha_C),
\end{align}
so $g^*=0$ for $\kappa<\bar{\kappa}^{(+)}$ and $g^*=1$ for $\kappa\ge \bar{\kappa}^{(+)}$ (contested common-interest state).

\smallskip
\noindent
(ii) If $\alpha_E<\theta$ and $\alpha_C>s$, then $\alpha_E s - \theta \alpha_C<0$ and the cutoff is
\begin{align}
\bar{\kappa}^{(-)}(\alpha_C,\alpha_E,\sigma)
= \frac{\theta-\alpha_E}{\theta \alpha_C - \alpha_E s}\in(0,1/\alpha_C),
\end{align}
so $g^*=0$ for $\kappa<\bar{\kappa}^{(-)}$ and $g^*=1$ for $\kappa\ge \bar{\kappa}^{(-)}$ (contested transfer state).

\smallskip
\noindent
On the knife-edge $\alpha_E s=\theta \alpha_C$, inequality \eqref{eq:key_ineq} reduces to $\alpha_E\gtreqless \theta$, and when also $\alpha_E=\theta$ the elite is indifferent between $g=0$ and $g=1$.
\end{proof}
%TC:endignore
% ============================
% Appendix: Cultural Evolution
% ============================

\section{Universalization vs.\ reciprocity}\label{app:reciprocity}

This appendix formalizes the difference between universalization 
and reciprocity-based explanations of state capacity. We retain 
the dynamic signaling environment of 
Section~\ref{sec:dynamics-aligned}, with two states 
$\alpha^L < \alpha^H$ and the elite privately informed, but 
replace Homo Moralis with a reciprocity preference. The 
comparison delivers one structural observation 
(Remark~\ref{rmk:recip-structure}) and one parametric ordering 
(Lemma~\ref{lem:recip-dominance}) of the within-period fiscal 
multiplier under the two behavioral foundations.

\subsection*{Reciprocity preferences}

Following \citet{besley2020}, a reciprocity citizen rewards the 
elite for the observed allocation $g$, with no inference about 
the underlying value of public goods $\alpha$. Let 
$\lambda \in [0,1)$ denote the strength of the reciprocity 
motive, and let $\psi:\{0,1\}\to\mathbb{R}_+$ encode the reward, 
with $\psi_1 \equiv \psi(1) \ge \psi_0 \equiv \psi(0)$. The 
structural feature defining reciprocity is that $\psi$ depends 
only on the observed action $g$, not on the inferred state 
$\alpha$ that the action might signal.\footnote{In 
\citet{besley2020}, reciprocity enters the civic-minded 
citizen's utility through the equilibrium gap $G-eB$, where $e$ 
is the elite's population share. This gap is positive when 
spending favors the public good and negative when it favors 
rents. We adopt the parametric $\psi(g)$ formulation to keep 
the comparison transparent, with the gap-based microfoundation 
nested as a special case. Besley conducts his analysis on the 
common-interest range $\alpha\in[1,A]$, a restriction on 
material incentives in his fiscal geometry rather than on where 
reciprocity is well-defined as a preference. The analog 
common-interest threshold in our environment is $\alpha>s$, 
where $s\equiv\sigma\theta(\sigma)$ is the share of diverted 
revenue recaptured by citizens as transfers.} In a separating equilibrium, observing 
$g=1$ is informationally equivalent to learning $\alpha=\alpha^H$, 
but under reciprocity this inference is payoff-irrelevant: the 
reward $\psi_1$ is a preference primitive, fixed across 
economies that differ only in $\alpha^H$. Allowing $\psi$ to 
respond to the inferred state would close the gap between the 
two frameworks by construction, converting reciprocity into 
universalization.

The citizen's optimal report follows the same functional form 
as Lemma~\ref{lem:optimal_report}, with the moral return 
$\varphi(g,\alpha,\sigma)$ replaced by $\psi(g)$:
\begin{equation}\label{eq:recip-report}
\tilde{w}^R(\lambda;g,t,c,\sigma) 
\;=\; w + \frac{t}{c}\bigl[\lambda\psi(g) - 1\bigr],
\end{equation}
yielding Laffer-maximizing per-capita revenue
\begin{equation}\label{eq:recip-Laffer}
% R2 (w=1). Original: \frac{wc}{4[1-\lambda\psi(g)]}
\hat{T}^R(g) \;=\; \rev{\frac{c}{4\bigl[1-\lambda\psi(g)\bigr]}},
\end{equation}
which is invariant to citizens' beliefs about $\alpha$. The 
elite's per-period payoff gain from provision relative to rents 
is
\begin{equation}\label{eq:recip-Delta}
% R2 (w=1). Original prefactor: \frac{wc}{4}
\Delta^R(\alpha;\lambda)
\;=\; \rev{\frac{c}{4}}\!\left[\frac{\alpha}{1-\lambda\psi_1}
- \frac{\theta(\sigma)}{1-\lambda\psi_0}\right].
\end{equation}
Because $\hat{T}^R(g)$ is belief-independent, $\Delta^R$ does 
not depend on citizens' posteriors. The incentive constraints 
for separation reduce to 
$\Delta^R(\alpha^H;\lambda) > 0 > \Delta^R(\alpha^L;\lambda)$, 
with no role for off-path beliefs or equilibrium refinements.

\subsection*{Structural separation}

The substantive comparison concerns the within-period fiscal 
multiplier, the factor by which credible provision raises the 
tax base relative to rents. Under universalization, this 
multiplier is
\[
J^U(\kappa;\sigma) 
\;\equiv\; \frac{\hat{T}(1)}{\hat{T}(0)}
\;=\; \frac{1-\kappa s}{1-\kappa\alpha^H},
\]
where the numerator uses $\varphi(0,\alpha,\sigma)=s$ under 
rents and the denominator uses $\varphi(1,\alpha^H,\sigma)=\alpha^H$ 
under provision in the separating equilibrium. Under 
reciprocity, from~\eqref{eq:recip-Laffer},
\[
J^R(\lambda;\sigma) 
\;\equiv\; \frac{\hat{T}^R(1)}{\hat{T}^R(0)}
\;=\; \frac{1-\lambda\psi_0}{1-\lambda\psi_1}.
\]

\begin{remark}[Structural separation]\label{rmk:recip-structure}
Fix $(\psi_0, \psi_1, \lambda)$ and $(\kappa, \sigma)$, and 
compare economies that differ only in $\alpha^H$. For any 
reciprocity preference with $\psi_0 \le \psi_1 < 1/\lambda$, 
the multiplier $J^R$ is invariant to $\alpha^H$. The 
universalization multiplier $J^U$ is strictly increasing in 
$\alpha^H$, with
\[
\frac{\partial J^U}{\partial \alpha^H} 
\;=\; \frac{\kappa\,(1-\kappa s)}{(1-\kappa\alpha^H)^2} 
\;>\; 0.
\]
\end{remark}

The two frameworks make different predictions about the fiscal 
response to a shift in the high-state value of public goods. 
Two economies with identical institutions and identical 
reciprocity strength but different $\alpha^H$ generate identical 
multipliers under reciprocity and different multipliers under 
universalization. This separation is parametrization-invariant: 
it holds for every admissible $(\psi_0,\psi_1)$ and does not 
depend on matching the two models' preference parameters. The 
economic content is that reciprocity rewards the observed 
action, while universalization rewards the inferred state. In 
a signaling environment, provision is informative about 
$\alpha$, and the channel through which credible reform raises 
fiscal capacity runs through citizens' inference. Reciprocity 
severs this channel by construction.

\subsection*{Parametric dominance}

Within a matched calibration, the multiplier comparison admits 
a sharp characterization. The institutional baseline 
$s\equiv\sigma\theta(\sigma)$ appears in both models: under 
universalization as the moral return to compliance under rents, 
under reciprocity as the share of diverted revenue recaptured 
by citizens through institutional cohesion. Setting 
$\psi_0 = s$ matches the two frameworks on this common 
environmental anchor, and $\lambda = \kappa$ equates the 
strength of the moral motive across specifications.

\begin{lemma}[Parametric dominance]\label{lem:recip-dominance}
Under the normalization $\psi_0 = s$, $\lambda = \kappa$, and 
with $\alpha^H \in [s,1/\kappa)$:
\begin{enumerate}
\item For $\psi_1 \in [s,\alpha^H)$: 
$J^U(\kappa;\sigma) > J^R(\kappa;\sigma)$.
\item At $\psi_1 = \alpha^H$: $J^U = J^R$.
\item For $\psi_1 \in (\alpha^H,1/\kappa)$: $J^R > J^U$.
\end{enumerate}
\end{lemma}

\begin{proof}
With $\psi_0 = s$ and $\lambda = \kappa$, $J^U > J^R$ iff 
$1-\kappa\alpha^H < 1-\kappa\psi_1$, i.e.\ $\psi_1 < \alpha^H$. 
The boundary and reverse cases follow by inspection.
\end{proof}

The boundary $\psi_1 = \alpha^H$ has a transparent 
interpretation: it is the unique reciprocity calibration that 
exactly matches the universalization multiplier, requiring 
citizens to reward provision \emph{as if} they had inferred 
the high state. Any reciprocity preference that does not 
implicitly perform this inference, i.e.\ any 
$\psi_1 < \alpha^H$, sits in the universalization-dominated 
region. The reverse-dominance region $\psi_1 > \alpha^H$ 
requires citizens to over-respond relative to any rational 
inference about $\alpha$, which has no natural microfoundation.

\paragraph*{Equilibrium refinements.}
The D1 refinement required under universalization 
(\ref{app:D1}) plays no role under reciprocity: the incentive 
constraints \eqref{eq:recip-Delta} are belief-independent, so 
off-path beliefs cannot affect the elite's payoff. The 
signaling complications of the universalization model are a 
substantive consequence of the belief channel rather than an 
artifact of refinement.

\section{Cultural evolution of moral universalization}
\label{app:cultural}
This appendix supplies the formal derivations and proofs for 
the endogenous-morality extension sketched in 
Section~\ref{sec:robustness}, in the spirit of cultural evolution 
models \citep{sandholm2010, sethi2001, besley2020}. Like 
\citet{besley2020}, we embed cultural dynamics into a fiscal 
model by allowing the population share of a behavioral type to 
evolve according to relative experienced utilities. Two differences are that 
the evolving trait is moral universalization rather than 
reciprocity, and that the within-period game involves 
asymmetric information and signaling, making the equilibrium 
allocation $g^*(\mu)$ depend on the cultural state.

\paragraph*{Timing} Each period $n$ proceeds as follows.
\begin{enumerate}
\item The period opens with the universalizer share $\mu_n\in[0,1]$ 
inherited from period $n-1$, so aggregate morality is 
$\kappa_n \equiv \mu_n\kappa_H$.
\item Nature draws $\alpha_n\in\{\alpha^{L},\alpha^{H}\}$ from 
the Markov chain of Definition~\ref{def:dynamic-game}. The elite 
privately observes $\alpha_n$.
\item The within-period signaling game of 
Section~\ref{sec:dynamics-aligned} is played to completion: the 
elite chooses $g_n$, citizens form posterior $p_n$ and choose 
reports $\tilde{w}_n$.
\item Per-period payoffs $U_E(\alpha_n;g_n\mid p_n)$ and 
$U^{(\kappa_J)}(\tilde{w}_n\mid p_n)$ are realized for the elite 
and each citizen type $J\in\{M,H\}$.
\item Citizens are socialized: $\mu_{n+1}$ is determined by a 
revision protocol that favors the type with higher realized 
utility, yielding the dynamic specified in 
equation~\eqref{eq:mu-dynamics} below.\footnote{What replicates is 
a cultural trait shaped by experienced utility, socialization, 
and peer influence. Under material fitness, materialists always 
dominate by revealed preference since they optimize exactly the 
criterion governing selection.}
\end{enumerate}

The two scales are well separated. Within any period $\mu_n$ is 
fixed and the stationary equilibria of 
Section~\ref{sec:dynamics-aligned} characterize behavior at 
$\kappa_n = \mu_n\kappa_H$.\footnote{This treats the elite as 
myopic: it maximizes its per-period payoff taking $\mu_n$ as 
given, consistent with the myopia property of the revision 
protocol. A forward-looking elite would face a richer dynamic 
problem in which provision today raises the future universalizer 
share and hence future fiscal capacity. The stationary 
characterization serves as a tractable benchmark for that 
extension.} Across periods, $\mu_n$ evolves 
gradually toward an interior steady state. The goal is not to 
replace the main analysis, but to show how a positive level of 
moral universalization can be sustained endogenously through 
socialization and peer effects, and how the long-run fiscal 
regime is determined by the interaction between the cultural 
steady state and the within-period game.
\paragraph*{Types and aggregate moral universalization}\label{app:cultural:types}
The population of citizens consists of two cultural types:
\begin{itemize}
  \item \emph{Materialists} $M$, with moral universalization parameter $\kappa=0$.
  \item \emph{Universalizers} $H$, with moral universalization parameter $\kappa=\kappa_H>0$.
\end{itemize}
Since the optimal report in Lemma~\ref{lem:optimal_report} is 
linear in $\kappa$ at interior solutions (guaranteed by 
Assumption~\ref{ass:kappa}), compliance decisions aggregate 
linearly across types and $\kappa_n = \mu_n\kappa_H$ (step~1 
of the timing) is a sufficient statistic for the tax base.
\paragraph*{Policy environment and per-period utilities}\label{app:cultural:payoffs}
Fix the institutional environment $(c,\sigma)$ and let $s\equiv \sigma \theta(\sigma)$. Under the D1-refined equilibrium of Section~\ref{sec:dynamics-aligned}, citizens' posteriors equal realized states, so the Laffer-maximizing per-capita revenue at cultural state $\mu_n$, fundamentals $\alpha_n$, and policy $g_n\in\{0,1\}$ is
\begin{equation}\label{eq:Tmu}
T(g_n;\alpha_n,\mu_n)
\equiv g_n\,T_1(\alpha_n,\mu_n) + (1-g_n)\,T_0(\mu_n),
\end{equation}
with
\begin{equation}\label{eq:T10mu}
% R2 (w=1). Original numerators: \frac{wc}{4(\cdots)}
T_1(\alpha,\mu) \equiv \rev{\frac{c}{4\left(1-\mu\,\kappa_H\,\alpha\right)}},
\qquad
T_0(\mu) \equiv \rev{\frac{c}{4\left(1-\mu\,\kappa_H\,s\right)}}.
\end{equation}
Let $U^{(\kappa_H)}(g_n,\alpha_n,\mu_n)$ and $U^{(0)}(g_n,\alpha_n,\mu_n)$ denote the per-period utility of a universalizer and a materialist respectively, under policy $g_n$, fundamentals $\alpha_n$, and population share $\mu_n$.

Given the elite's equilibrium policy rule $g^*(\alpha,\mu)$, define the \emph{expected} indirect utility of each type at share $\mu$, integrating over the prior:
\begin{align}
U^{(\kappa_H)}(\mu) &\equiv \rho\,U^{(\kappa_H)}\!\left(g^*(\alpha^H,\mu),\,\alpha^H,\,\mu\right) + (1-\rho)\,U^{(\kappa_H)}\!\left(g^*(\alpha^L,\mu),\,\alpha^L,\,\mu\right), \\
U^{(0)}(\mu) &\equiv \rho\,U^{(0)}\!\left(g^*(\alpha^H,\mu),\,\alpha^H,\,\mu\right) + (1-\rho)\,U^{(0)}\!\left(g^*(\alpha^L,\mu),\,\alpha^L,\,\mu\right),
\end{align}
where $\rho = \Pr(\alpha = \alpha^H)$.

\paragraph*{Cultural fitness advantage}\label{app:cultural:advantage}
Define the \emph{state-conditional} cultural fitness advantage 
of universalizers as the per-period utility difference
\begin{equation}\label{eq:cultural-advantage}
\overline{\Delta}_C(g;\mu\mid\alpha) \equiv 
U^{(\kappa_H)}\!\left(g,\alpha,\mu\right) - 
U^{(0)}\!\left(g,\alpha,\mu\right),
\end{equation}
and the unconditional advantage by integrating over the prior:
\begin{equation}\label{eq:cultural-advantage-uncond}
\overline{\Delta}_C(g;\mu) \equiv 
\rho\,\overline{\Delta}_C(g;\mu\mid\alpha^H) + 
(1-\rho)\,\overline{\Delta}_C(g;\mu\mid\alpha^L).
\end{equation}
The sign and monotonicity of 
$\overline{\Delta}_C(g;\mu\mid\alpha)$ determine whether moral 
universalization persists, vanishes, or converges to an interior 
steady state. Because the sign turns out to be independent of 
$\alpha$, the unconditional advantage $\overline{\Delta}_C(g;\mu)$ 
that drives the replicator inherits it. We derive a closed-form 
expression below and summarize the key properties here.
\begin{proposition}[Cultural fitness advantage at fixed allocation]
\label{prop:cultural-fitness}
Fix $g\in\{0,1\}$ and $\alpha\in\{\alpha^L,\alpha^H\}$, and define $A(g) \equiv \kappa_H\,\varphi(g,\alpha,\sigma)$. Under Assumption~\ref{ass:kappa} applied to the maximum moral type, $\kappa_H < 1/\alpha^H$, we have $A(g)<1$, so the expressions below are well-defined for all $\mu\in[0,1]$. The state-conditional cultural fitness advantage of universalizers is
\begin{equation}\label{eq:Delta-closed}
% R2 (w=1). Original prefactor: \frac{wc}{8}
\overline{\Delta}_C(g;\mu\mid\alpha) = \rev{\frac{c}{8}}\,
\frac{A(g)^2\,(1-2\mu)}{(1-\mu\, A(g))^2},
\end{equation}
which satisfies:
\begin{enumerate}
\item $\overline{\Delta}_C(g;\mu\mid\alpha) > 0$ for $\mu < 1/2$, 
\;$\overline{\Delta}_C(g;\mu\mid\alpha) = 0$ for $\mu = 1/2$, 
\;$\overline{\Delta}_C(g;\mu\mid\alpha) < 0$ for $\mu > 1/2$;
\item $\frac{\partial}{\partial\mu}\overline{\Delta}_C(g;\mu\mid\alpha) < 0$ 
for all $\mu \in [0,1]$: the advantage is strictly decreasing 
in $\mu$.
\end{enumerate}
Since the sign and monotonicity are independent of $\alpha$, 
the same properties hold for the unconditional advantage 
$\overline{\Delta}_C(g;\mu)$ defined in~\eqref{eq:cultural-advantage-uncond}.
\end{proposition}

Since $\varphi(g,\alpha,\sigma)\le\alpha^H$ for all $g$ and all realized $\alpha$, the condition $\kappa_H<1/\alpha^H$ implies $\mu\kappa_H\varphi(g,\alpha,\sigma)<1$ for every $\mu\in[0,1]$. The aggregate moral level $\kappa_n = \mu_n\kappa_H$ therefore satisfies Assumption~\ref{ass:kappa} throughout the cultural transition, and the within-period equilibrium is well-defined at every cultural state.

The economic intuition is as follows.
When universalizers are rare ($\mu$ small), the Laffer-maximizing 
tax rate is moderate and the moral return from compliance exceeds 
the private cost of over-reporting relative to materialists.
Universalizers enjoy a net utility advantage from their moral 
satisfaction.
As universalizers become prevalent, aggregate compliance rises, 
which pushes the Laffer-maximizing tax rate higher.
Both types bear this higher tax burden, but universalizers, who 
report more than materialists, bear it disproportionately.
At the same time, the moral return from compliance is partially 
capitalized into the fiscal system, reducing the marginal benefit 
of further moral behavior.
The crossover occurs exactly at $\mu = 1/2$: below this 
threshold, universalizers are advantaged. Above it, materialists 
are.

\paragraph*{A generic cultural evolutionary dynamic}
\label{app:cultural:dynamic}
We adopt a standard revision-protocol representation in discrete 
time \citep{sandholm2010}. Two properties characterize this 
class of dynamics: \emph{inertia}, whereby agents revise their 
cultural type only sporadically rather than continuously, and 
\emph{myopia}, whereby revision decisions are conditioned on 
current payoffs rather than forward-looking expectations. 
For brevity, let $M$ denote the materialist type ($\kappa=0$) 
and $H$ the universalizer type ($\kappa=\kappa_H$). 
Let $s_{MH}(\mu)$ be the switching rate from $M$ to $H$ and 
$s_{HM}(\mu)$ the switching rate from $H$ to $M$. The cultural 
composition evolves according to
\begin{equation}\label{eq:mu-dynamics}
\mu_{n+1}-\mu_n = (1-\mu_n)\,s_{MH}(\mu_n)
-\mu_n\,s_{HM}(\mu_n).
\end{equation}
We impose only the sign restriction that switches favor the 
type with higher cultural fitness:
\begin{equation}\label{eq:revision-sign}
s_{MH}(\mu)>0 \ \Longleftrightarrow\ 
\overline{\Delta}_C(g;\mu)>0,
\qquad
s_{HM}(\mu)>0 \ \Longleftrightarrow\ 
\overline{\Delta}_C(g;\mu)<0.
\end{equation}
We work with the discrete-time replicator
\begin{equation}\label{eq:replicator}
\mu_{n+1}-\mu_n = \eta\,\mu_n(1-\mu_n)\,
\overline{\Delta}_C(g;\mu_n), \qquad \eta>0,
\end{equation}
a standard special case of the revision protocol 
\eqref{eq:mu-dynamics}. The replicator \eqref{eq:replicator} 
applies equally when $g$ is constant or given by the endogenous 
equilibrium rule $g^*(\mu_n)$ of the within-period signaling 
game, with $\overline{\Delta}_C$ evaluated at the prevailing 
allocation. Linearity of the optimal report in $\kappa$ 
(Lemma~\ref{lem:optimal_report}) ensures that the elite's 
incentive constraints at cultural state $\mu$ take the same 
form as in Propositions~\ref{prop:weak-high}--\ref{prop:strong-high} 
at $\kappa_n=\mu_n\kappa_H$, so $g^*(\mu)$ is well-defined for 
each $\mu\in(0,1)$.

\begin{proposition}[Convergence to interior morality]
\label{prop:cultural-convergence}
Maintain Assumption~\ref{ass:kappa} at the maximum moral type ($\kappa_H<1/\alpha^H$), and let $\kappa^{\mathrm{pr}}\equiv\bar\kappa(\alpha,\sigma)$ denote the provision threshold from Proposition~\ref{prop:Elite}. Suppose further $\kappa_H/2\neq\kappa^{\mathrm{pr}}$ (a generic condition that excludes a knife-edge case). Then under the replicator \eqref{eq:replicator} with $g=g^*(\mu_n)$:
\begin{enumerate}
\item The interior point $\mu^*=1/2$ is globally asymptotically 
stable on $(0,1)$, with aggregate morality 
$\kappa^*=\kappa_H/2>0$. The boundary states $\mu=0$ and 
$\mu=1$ are unstable.
\item The long-run fiscal regime is $g^*(\mu^*)=1$ if 
$\kappa_H/2>\kappa^{\mathrm{pr}}$ and $g^*(\mu^*)=0$ 
otherwise.
\end{enumerate}
\end{proposition}

\begin{figure}[!htb]
\centering
\includegraphics{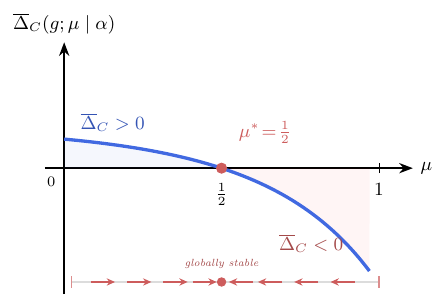}
\caption{Cultural fitness advantage and convergence to interior morality.
The curve plots the state-conditional advantage 
$\overline{\Delta}_C(g;\mu\mid\alpha)$ of 
Proposition~\ref{prop:cultural-fitness} at an illustrative 
$A(g)=\kappa_H\varphi(g,\alpha,\sigma)$. It is strictly concave 
on $[0,1]$ and asymmetric about $\mu=1/2$, with a shallow positive 
region and a steep descent as $\mu\to 1$ driven by the 
$(1-\mu A)^{-2}$ factor. For $\mu < 1/2$, 
universalizers earn higher utility and their share grows. For 
$\mu > 1/2$, materialists are favored. The unique interior 
steady state $\mu^* = 1/2$ is globally asymptotically stable 
(Proposition~\ref{prop:cultural-convergence}). Arrows indicate 
the direction of the replicator dynamic~\eqref{eq:replicator}.}
\label{fig:cultural-phase}
\end{figure}

\begin{proof}
By Proposition~\ref{prop:cultural-fitness}, $\overline{\Delta}_C(g;\mu\mid\alpha)=\rev{\frac{c}{8}}\,\frac{A(g)^2(1-2\mu)}{(1-\mu A(g))^2}$ has the same sign as $1-2\mu$ for every $(g,\alpha)$. Integrating over $\alpha$ with weights $(\rho,1-\rho)$ preserves the sign, so the expected advantage driving \eqref{eq:replicator} also has the same sign as $1-2\mu$. A switch in $g^*(\mu)$ at the policy threshold $\mu^{\mathrm{switch}}\equiv\kappa^{\mathrm{pr}}/\kappa_H$ therefore does not alter the direction of motion of $\mu_n$, and the replicator \eqref{eq:replicator} evaluated at $g=g^*(\mu_n)$ satisfies $\mu_{n+1}>\mu_n$ for $\mu_n\in(0,1/2)$ and $\mu_{n+1}<\mu_n$ for $\mu_n\in(1/2,1)$. Hence $\mu^*=1/2$ is globally attracting on $(0,1)$.\footnote{For $\eta$ small enough that the step never overshoots $\mu^*=1/2$ (i.e., $\mu_{n+1}\le 1/2$ whenever $\mu_n\le 1/2$, and $\mu_{n+1}\ge 1/2$ whenever $\mu_n\ge 1/2$), convergence is monotone. This regime is consistent with the inertia property of revision protocols \citep{sandholm2010} and we maintain it throughout.} At $\mu=0$ (resp.\ $\mu=1$), $\overline{\Delta}_C(g;0)>0$ (resp.\ $\overline{\Delta}_C(g;1)<0$), so a small perturbation toward the interior grows. This proves part~(1).

Part~(2) follows from Proposition~\ref{prop:Elite}: at 
$\kappa^*=\kappa_H/2$, $g^*=1$ if 
$\kappa_H/2>\kappa^{\mathrm{pr}}$ and $g^*=0$ otherwise.
The knife-edge $\kappa_H/2=\kappa^{\mathrm{pr}}$ is 
excluded by assumption.
\end{proof}

\begin{corollary}[Long-run fiscal regime]
\label{cor:endogenous-credibility}
Let $\mu^* = 1/2$ be the steady state of Proposition~\ref{prop:cultural-convergence}, so $\kappa^* = \kappa_H/2$. Applying Proposition~\ref{prop:weak-high} at $\kappa^*$, the long-run fiscal regime is:
\begin{enumerate}
\item \emph{Pooling at rents} ($g^*\equiv 0$) if $\kappa_H/2 < \kappa_{\min}^H$;
\item \emph{D1-refined separation} ($g^*(\alpha^H)=1$, $g^*(\alpha^L)=0$) if $\kappa_H/2\in[\kappa_{\min}^H,\kappa_{\max}^L)$;
\item \emph{Pooling at provision} ($g^*\equiv 1$), feasibility permitting,\footnote{The additional conditions $\alpha^H\sigma\le\bar{\alpha}$ and $\kappa_H/2\ge\max\{\kappa^{\text{pool}},\kappa^{H,\min}\}$ from Proposition~\ref{prop:weak-high} must also hold.} if $\kappa_H/2\ge\kappa_{\max}^L$.
\end{enumerate}
In case (2), credible reform persists without exogenous maintenance of morality.
\end{corollary}

The three cases correspond to the three regions of the within-period game at $\kappa=\kappa^*$. Which one obtains depends on the \emph{level} of the maximum moral type $\kappa_H$ relative to the separating region of Proposition~\ref{prop:weak-high}. Cultural evolution fixes the share $\mu^*=1/2$ but not the level.

\paragraph*{Scope of the convergence result}
The convergence to $\mu^* = 1/2$ does not depend on the within-period signaling structure. Because $\overline{\Delta}_C(g;\mu\mid\alpha)$ inherits the sign of $(1-2\mu)$ for every $(g,\alpha)$, the moral steady state is the same whether the within-period game involves D1 separation, pooling, or no signaling at all: the same $\mu^*$ obtains under complete information about $\alpha$. Signaling enters the appendix only through the endogenous policy rule $g^*(\mu)$ that appears in the replicator~\eqref{eq:replicator}. It does not affect the location or stability of $\mu^*$. The role of signaling is rather to pin down the long-run \emph{fiscal} regime: at $\kappa^* = \kappa_H/2$, the within-period equilibrium maps morality into one of the three regions of Corollary~\ref{cor:endogenous-credibility}. Cultural dynamics therefore deliver the moral state. The within-period game allocates it between provision and rents.

The result establishes that moral universalization persists at 
a positive level as the generic long-run outcome of cultural 
evolution. Unlike the tipping dynamics in \citet{besley2020}, 
where convergence to a high- or low-capacity steady state 
depends on initial conditions, the convergence result here is 
global: any interior initial condition leads to $\mu^* = 1/2$ 
and hence to $\kappa^* = \kappa_H/2 > 0$. This provides a 
cultural-evolutionary rationale for the assumption $\kappa > 0$ 
maintained throughout the main analysis.
\subsection{Proof of Proposition~\ref{prop:cultural-fitness}}
\label{app:cultural:delta}
Fix $g\in\{0,1\}$ and let $\varphi_g \equiv \varphi(g,\alpha,\sigma)$, 
$A(g) \equiv \kappa_H\,\varphi_g$, and 
$D(\mu) \equiv 1 - \mu\,A(g)$.
% R2: removed redundant local normalization and the erroneous "scale linearly" claim (revenue scales as w^2). Normalization is now global; see Section~\ref{sec:baseline}.
\rev{Recall the normalization $w=1$ from Section~\ref{sec:baseline}.}

\paragraph{Step 1: Laffer-maximizing tax rate and revenue}
\begin{equation}
\hat{t}(g,\mu) = \frac{c/2}{D(\mu)},
\qquad
\hat{T}(g,\mu) = \frac{c}{4\,D(\mu)}.
\end{equation}

\paragraph{Step 2: Type-specific reports and deviations}
Using $\hat{t}/c = 1/(2D)$, the optimal reports are
\begin{equation}
\tilde{w}^M(g,\mu) = 1 - \frac{1}{2D(\mu)},
\qquad
\tilde{w}^H(g,\mu) = 1 + \frac{A(g)-1}{2D(\mu)},
\end{equation}
with deviations $d^M = -1/(2D)$ and $d^H = (A(g)-1)/(2D)$.

\paragraph{Step 3: Indirect utilities}
\begin{equation}
U^{(\kappa_J)}(g,\mu) = z^{(\kappa_J)}(g,\mu) + \varphi_g
\left[(1-\kappa_J)\,\hat{T}(g,\mu) +
\kappa_J\,\hat{t}(g,\mu)\,\tilde{w}^{(\kappa_J)}(g,\mu)
\right],
\end{equation}
where $z^{(\kappa_J)} = 1 - \hat{t}\,\tilde{w}^{(\kappa_J)}
- \frac{c}{2}(d^{(\kappa_J)})^2$.
\paragraph{Step 4: Cultural fitness advantage}
Substituting and simplifying yields the state-conditional advantage
\begin{equation}
\overline{\Delta}_C(g;\mu\mid\alpha) =
U^{(\kappa_H)}(g,\mu) - U^{(0)}(g,\mu) =
\frac{c}{8}\,\frac{A(g)^2\,(1-2\mu)}{D(\mu)^2},
\end{equation}
where the dependence on $\alpha$ enters through $A(g)=\kappa_H\,\varphi_g$.

\paragraph{Step 5: Monotonicity}
Differentiating with respect to $\mu$:
\begin{equation}
\frac{\partial}{\partial\mu}\overline{\Delta}_C(g;\mu\mid\alpha)
= \frac{c\,A(g)^2}{4}\,
\frac{1 - A(g)(1-\mu)}{(\mu\,A(g) - 1)^3}.
\end{equation}
Since $A(g) < 1$ by assumption, the numerator satisfies
$1 - A(g)(1-\mu) \ge 1 - A(g) > 0$ and the denominator
$(\mu\,A(g) - 1)^3 < 0$ for all $\mu\in[0,1]$, so
$\frac{\partial}{\partial\mu}\overline{\Delta}_C(g;\mu\mid\alpha) < 0$.

\paragraph{Step 6: Unique root}
$\overline{\Delta}_C(g;\mu\mid\alpha) = 0$ iff $1 - 2\mu = 0$,
i.e.\ $\mu^* = 1/2$. The unconditional advantage
$\overline{\Delta}_C(g;\mu)=\rho\,\overline{\Delta}_C(g;\mu\mid\alpha^H)
+(1-\rho)\,\overline{\Delta}_C(g;\mu\mid\alpha^L)$ inherits the
sign, monotonicity, and unique root.
\qed

\end{document}